\documentclass[final,twocolumn,5p,authoryear]{elsarticle}
\usepackage{url}

\usepackage{amsmath, amsfonts, amsthm}
\usepackage{mathalfa}
\usepackage[bitstream-charter]{mathdesign}
\usepackage{dsfont} 
\usepackage{empheq} 
\usepackage{color}
\usepackage{tikz}
\usepackage{caption}
\usepackage{subcaption}
\usepackage{enumitem}
\newlist{eprop}{enumerate}{1}
\setlist[eprop]{leftmargin=+4em, label = {\bf(EP-\arabic*)}, resume}
\newlist{ass}{enumerate}{1}
\setlist[ass]{leftmargin=+3em, label = {\bf(A-\arabic*)}, resume}

\newtheorem{Theorem}{Theorem}
\newtheorem{Definition}{Definition}

\newtheorem{Proposition}[Definition]{Proposition}
\newtheorem{Lemma}[Definition]{Lemma}
\newtheorem{Remark}[Definition]{Remark}

\newcommand{\D}{\mathbb{D}}

\newcommand{\N}{\mathbb{N}}
\newcommand{\M}{\mathbb{M}}

\newcommand{\R}{\mathbb{R}}

\newcommand{\Z}{\mathbb{Z}}

\newcommand{\cC}{\mathcal{C}}

\newcommand{\cI}{\mathcal{I}}

\newcommand{\cO}{\mathcal{O}}

\newcommand{\cQ}{\mathcal{Q}}

\newcommand{\cW}{\mathcal{W}}

\newcommand{\fC}{\mathfrak{C}}

\newcommand{\me}{\mathtt{e}}

\DeclareMathOperator{\rank}{rank}
\DeclareMathOperator{\im}{im}

\newcommand{\eps}{\varepsilon}
\newcommand{\setdef}[2]{\left\{\ #1\ \left|\ \vphantom{#1} #2\ \right.\right\}}

\newcommand{\pw}{\text{\normalfont pw}}

\newcommand{\cons}{\text{\normalfont cons}}
\newcommand{\imp}{\text{\normalfont imp}}
\newcommand{\diff}{\text{\normalfont diff}}

\newcommand{\du}{\mathtt{u}}
\newcommand{\dy}{\mathtt{y}}

\newcommand{\Dpwsm}{\D_{\pw\cC^\infty}}

\newenvironment{smallbmatrix}
{\left[\begin{smallmatrix}}
{\end{smallmatrix}\right]}

\newcommand\blfootnote[1]{%
  \begingroup
  \renewcommand\thefootnote{}\footnote{#1}%
  \addtocounter{footnote}{-1}%
  \endgroup
}
\begin{document}

\begin{frontmatter}
\title{Determinability and state estimation for switched differential-algebraic equations}

\author[T]{Aneel Tanwani}\ead{aneel.tanwani@laas.fr }
\author[S]{Stephan Trenn}\ead{trenn@mathematik.uni-kl.de }

\address[T]{LAAS--CNRS, University of Toulouse, 31400 Toulouse, France.}
\address[S]{Department of Mathematics, University of Kaiserslautern, 67663 Kaiserslautern, Germany.} 


\begin{abstract}
The problem of state reconstruction and estimation is considered for a class of switched dynamical systems whose subsystems are modeled using linear differential-algebraic equations (DAEs).
Since this system class imposes time-varying dynamic and static (in the form of algebraic constraints) relations on the evolution of state trajectories, an appropriate notion of observability is presented which accommodates these phenomena.
Based on this notion, we first derive a formula for the reconstruction of the state of the system where we explicitly obtain an injective mapping from the output to the state.
In practice, such a mapping may be difficult to realize numerically and hence a class of estimators is proposed which ensures that the state estimate converges asymptotically to the real state of the system.
\end{abstract}
\end{frontmatter}
\section{Introduction}
\blfootnote{This work was supported by DFG grant {TR 1223/2-1}.}
Switched differential-algebraic equations (DAEs) form an important class of switched systems, where the dynamics are not only discontinuous with respect to time, but also the state trajectories are constrained by certain algebraic equations which may also change as the system switches from one mode to another.
Such dynamical models have found utility e.g.\ in the analysis of electrical power distribution \citep{GrosTren14} and of electrical circuits \citep{Tren12}.
We consider switched DAEs of the following form
\begin{equation}\label{eq:sysLin}
\begin{aligned}
E_\sigma \dot x &= A_\sigma x + B_\sigma u \\
y &= C_\sigma x
\end{aligned}
\end{equation}
where $x, u , y$ denote the state (with dimension $n\in\N$), input (with dimension $\du\in\N$) and output (with dimension $\dy\in\N$) of the system,  respectively. The switching signal $\sigma: (t_0,\infty) \rightarrow \N$ is a piecewise constant, right-continuous function of time and in our notation it changes its value at time instants $t_1<t_2<\ldots$ called the switching times.
For a fixed $\sigma$, the triplet $(x,u,y)$ is used to denote signals satisfying \eqref{eq:sysLin}.
We adopt the convention that over the interval $[t_p, t_{p+1})$ of length $\tau_p:=t_{p+1}-t_p$, the active mode is defined by the quadruple $(E_p,A_p,B_p,C_p)\in\R^{n\times n} \times \R^{n\times n} \times \R^{n\times \du} \times \R^{\dy\times n}$, $p \in \N$.
Over the interval $(t_0,t_1)$, it is assumed that the system has some past which may be described by $(E_0,A_0,B_0,C_0)$.
The solution concept for switched DAEs is studied in \citep{Tren09d} and a brief discussion is also included in Section~\ref{sec:sol}.

Dynamical system with discontinuous, or constrained trajectories have gathered a lot of interest in the control community, as they form an important class of hybrid, or discontinuous dynamical systems, see e.g.\ \citep{GoebSanf09}.
One direction of research for these systems includes the study of structural properties that could be used for control design problems, and in this regard the problem of state reconstruction and estimation is of particular interest \citep{TanwPhD11}.
In the current literature, one finds that the earlier work on observability and observers for switched/hybrid systems was aimed at using classical Luenberger observers for continuous dynamics and treat the switching or discontinuity as a perturbation that can be overcome by certain strong assumptions on system data.
This line of work also requires that the underlying observability notion allows instantaneous recovery of the state from the measured output \citep{BabaPapp05,VidaChiu03}.
Modified Luenberger observers have also been used for constrained dynamical systems using similar observability concepts: see \citep{TanwBrog14} when the constraint sets are convex (or, mildly nonconvex), and \citep{TanwBrog16} when the constraints result in impacts.

However, for switched dynamical systems, several different notions of observability can be defined.
In \citep{SunGe02, XieWang03}, a switched system comprising ordinary differential equations (ODEs) is called observable if there exists a switching signal that allows reconstruction of the initial state from the output measurements over an interval.
This concept also appears in the observer construction (for continuous state and switching signal) proposed in \citep{BallBenv03}.
However, in our recent work, a more relaxed notion of observability has been proposed for switched ODEs \citep{TanwShim13} and switched DAEs \citep{TanwTren12}.
The switching signal in this case is assumed to be known and fixed (i.e.\ playing the same role as the input $u$ which is also assumed to be known and not influenced by the observer). By measuring the outputs and inputs over an interval, and using the data of subsystems activated during that interval, it is determined whether state reconstruction is possible or not.
Several variants of this notion are also collected in the survey \citep{PetrTanw15}.
A state estimation algorithm based on these generalized observability concept can be found in \citep{ShimTanw14, TanwShim13, TanwShim15} for switched ODEs.
The main contribution of this paper is to address the observer design for switched DAEs which, apart from our preliminary work \citep{TanwTren13}, has not been addressed in this literature.

Already in the context of nonswitched systems several challenges arise in the study of DAEs.
The difference basically arises due to the presence of algebraic equations (static relations) in the description of the system because of which the state trajectories can only evolve on the sets defined by the algebraic equations of the active mode.
Observer designs have been studied for (nonswitched) DAEs since 1980's, e.g. \citep{Dai89a,FahmOrei89}.
Unlike ODEs, the observer design in DAEs requires additional structural assumptions and, furthermore, the order of the observer may depend on the design method.
Because of these added generalities, observer design for nonswitched DAEs is still an active research field \citep{BobiCamp11, Daro12}, and the recent survey articles summarize the development of this field \citep{BergReis15ppb, BobiCamp14}.

In studying switched DAEs, our modeling framework allows for time-varying algebraic relations.
The changes in algebraic constraints due to switching introduce jumps in the state of the system, and because of the possibility of a higher-index DAE, these jumps may get differentiated and generate impulsive solutions.
The notion of observability studied in this paper takes into account the additional structure due to algebraic constraints, and the added information from the outputs in case there are impulses observed in the measurements.
This observability concept is then used to construct a mapping from the output space to the state space, which allows us to theoretically reconstruct the state.
The key element of constructing this mapping is to show how the structure of a linear DAE is exploited to recover the information about the state in individual subsystems.
This structural decomposition is then combined with the expressions used for evolution of states in switched DAEs to accumulate all possible recoverable information from past measurements about the state at one time instant.
The construction then yields a systematic procedure for writing the value of the state at a time instant in terms of outputs measured over an interval for which the system is observable.

The theoretical mappings constructed in the process are often not realizable in practice, but the derivation is used to describe a general class of state observers that generate asymptotically converging state estimates.
The main result states that if the observable components of the individual subsystems can be estimated well-enough, and the required observability assumption persistently holds with time, then the estimates converge asymptotically.
Our initial work on observers for switched DAEs \citep{TanwTren13}, and even the observers proposed for switched ODEs \citep{ShimTanw14, TanwShim13} can be seen as a special case of the general class of state estimators studied in this paper.


\section{Contribution and Layout}

This section provides a summary of all the technical results that are developed in this article, and a coherent view of how the different components are connected together to solve the state estimation problem for switched DAEs. The reader may also refer to this section as an index for finding the appropriate section for technical terms.
The main contribution of this article is to provide a systematic procedure for designing observers for system class~\eqref{eq:sysLin} which generate asymptotically convergent state estimates.
The flow-diagram which shows the working of the proposed observer is given in Figure~\ref{fig:obsAll}.

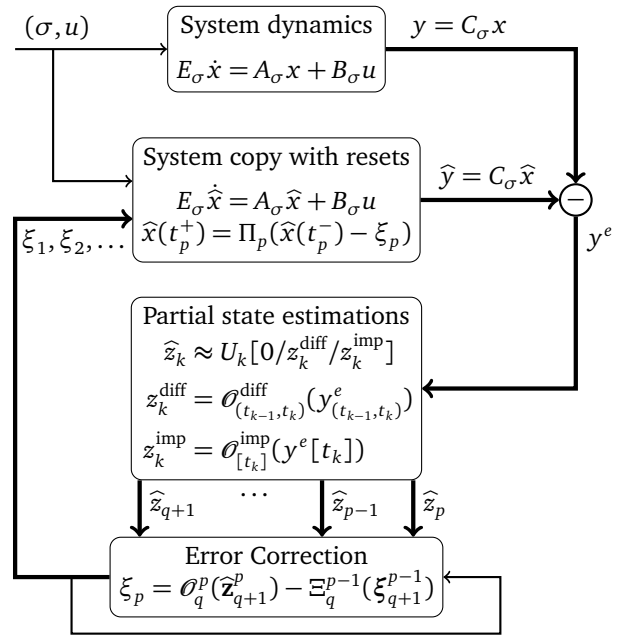
\begin{figure}[hbt]
\centering
\begin{tikzpicture}
\draw (0,0) node [rectangle, rounded corners, draw, align=center, minimum height =0.65cm, text centered] (sys) {System dynamics\\[0.4\baselineskip] $E_\sigma \dot x = A_\sigma x + B_\sigma u$};
\draw (0,-2) node [rectangle, rounded corners, draw, align=center, minimum height =0.65cm, text centered] (copy) {System copy with resets \\[0.4\baselineskip] $E_\sigma \dot{\widehat x} = A_\sigma \widehat x + B_\sigma u$\\ $\widehat x(t_p^+) = \Pi_p(\widehat x(t_p^-) - \xi_p)$};
\draw (0,-4.5) node [rectangle, rounded corners, draw, align=center, minimum height =0.65cm] (estimate) {Partial state estimations\\[0.5ex] $\begin{aligned}
  \widehat{z}_k &\approx  U_k[0/z^\diff_k/z_k^\imp]\\[0.2ex]
  z^\diff_k &= \cO_{(t_{k-1},t_k)}^\diff(y^e_{(t_{k-1},t_k)})\\
  z^\imp_k &= \cO_{[t_k]}^\imp (y^e[t_k])
  \end{aligned}$};

\draw (0,-7) node [rectangle, rounded corners, draw, align=center, minimum height =0.65cm] (corr) {Error Correction\\ $\xi_p = \cO_{q}^p(\widehat{\mathbf{z}}_{q+1}^p) - \Xi_{q}^{p-1}(\boldsymbol{\xi}_{q+1}^{p-1})$};
\coordinate (tr) at ([xshift=2.5cm)]sys.east);
\coordinate (mr) at ([yshift=-2cm)]tr);
\coordinate (br) at ([yshift=-4cm)]tr);
\coordinate (tl) at ([xshift=-2cm)]sys.west);
\coordinate (tml) at ([xshift=0.5cm)]tl);
\coordinate (ml) at ([yshift=-2cm)]tml);
\coordinate (bl) at ([yshift=-4cm)]tml);
\draw [thick,black] (mr) node [draw, circle,inner sep = 0.25mm] (sum) {$-$};
\draw [ultra thick, ->] (sys.east) node[anchor = south west] {$\ \ y = C_\sigma x$} -- (tr) -- (sum.north);
\draw [ultra thick, ->] (copy.east) node[anchor = south west]{$\ \widehat y = C_\sigma \widehat x$} -- (sum.west);
\draw [ultra thick, ->] (sum.south) node[anchor = north west]{$ y^e$} |- (estimate.east);
\draw [thick, ->] (tl) node[anchor = south west] {$(\sigma,u)$} -- (sys.west);
\draw [thick, ->] (tml) -- ([yshift=0.25cm]ml) -- ([yshift=0.25cm]copy.west); 
\draw [ultra thick, ->] (corr.west) -| ([xshift = -0.5cm, yshift=-0.25cm]ml) -- ([yshift=-0.25cm]copy.west) node [anchor=north east] {$\xi_1, \xi_2, \dots\!$}; 
\draw [thick, ->] ([xshift=-0.5cm]corr.west) -- ([xshift=-0.5cm, yshift = -0.75cm]corr.west) -- ([xshift=0.75cm, yshift = -0.75cm]corr.east) -- ([xshift=0.75cm]corr.east) -- (corr.east);
\draw[ultra thick,->] ([xshift=-1.8cm]estimate.south) node[anchor = north west] {$\widehat{z}_{q+1}$} -- ([xshift=-1.8cm]corr.north);
\draw[ultra thick,->] ([xshift=1.8cm]estimate.south) node[anchor = north west] {$\widehat{z}_{p}$} -- ([xshift=1.8cm]corr.north);
\draw[ultra thick,->] ([xshift=0.6cm]estimate.south) node[anchor = north west] {$\widehat{z}_{p-1}$} -- ([xshift=0.6cm]corr.north);
\draw ([xshift=-0.6cm]estimate.south) node[anchor = north west] {$\cdots$};
\end{tikzpicture}
\caption{A schematic representation of the state estimator.}
\label{fig:obsAll}
\end{figure}

Our proposed observer basically comprises two layers. In the first layer, a system copy with variable $\widehat x$ (also the state estimate) is simulated with same time scale as the actual plant, and $\widehat x$ is reset at some switching times $t_p$ by an error correction vector $\xi_p$.
The second layer comprises an algorithm to compute $\xi_p$ in very short (negligibly small) time. This article (in Sections~\ref{sec:obsCond}--\ref{sec:obsDesign}) develops the theoretical tools and numerical certificates that gurantee efficient implementation of this algorithm. The central idea is that $\xi_p$ closely approximates the estimation error prior to the reset times $e(t_p^-):=\widehat x(t_p^-) - x(t_p^-)$ by using the output measurements over the interval $(t_q,t_p]$, $q<p$.
If this approximation can be achieved to a desired accuracy, and these sufficiently rich approximates are injected into the estimate's dynamics by resetting $\widehat x$ often enough, then our algorithms ensure that $\widehat x(t)$ converges to $x(t)$ as $t$ tends to infinity.
There are three primary mechanisms involved in achieving this desired goal:
\begin{itemize}[leftmargin=*]
\item Determinability notion in Section~\ref{sec:obsCond} captures the property that, by measuring the output of \eqref{eq:sysLin} over an interval $(t_q,t_p]$, we must be able to recover the state of the system $x(t_p^+)$ exactly; and Theorem~\ref{thm:DAEDetMult} provides conditions on system data to achieve this property.
\item In Section~\ref{sec:outMaps}, under the assumption that system~\eqref{eq:sysLin} is determinable over the interval $(t_q,t_p]$, we build the theoretical engine to compute $\xi_p$.
This involves looking at the observable components of the individual subsystems for the error dynamics denoted by $z_k$, $k = q+1, \cdots, p$. Using the structure of a DAE, we define the maps $\cO_k^\diff$, $\cO_k^\imp$, and decompose $z_k$ into $z_k^\cons$, $z_k^\diff := \cO_k^\diff(y^e_{(t_{k-1},t_k]})$, $z_k^\imp := \cO_k^\imp(y^e_{[t_k]})$ based on whether they can be recovered from algebraic constraints, smooth output measurements, or impulsive output measurements, respectively.
Theorem~\ref{thm:mapZs} then provides the construction of a matrix $\cO_q^p$ such that
\[
e(t_p^+) = \Pi_p\cO_q^p \mathbf{z}_{q+1}^p,
\]
where $\mathbf{z}_{q+1}^p:=(z_{q+1}, \dots, z_p)$, and $\Pi_p$ is a projector depending on the system data $(E_p,A_p)$.

\item In Section~\ref{sec:obsDesign}, we show how to construct the estimate $\widehat z_k$ for the observable component of the error dynamics. These estimates are then used to define the vector $\xi_p$ recursively (see Fig.~\ref{fig:obsAll} also). The final result, Theorem~\ref{thm:obsGenConv}, proves the asymptotic convergence of $\widehat x(t)$ toward $x(t)$, as $t \rightarrow \infty$, in a rigorous manner.
\end{itemize}

\section{Preliminaries}\label{sec:prelim}

Before addressing the problem of observability, or state estimation, it helps to recall the solution framework associated with \eqref{eq:sysLin}, and certain algebraic properties of the system matrices, that will form the foundation of the analysis carried out in this paper.
In Section~\ref{sec:sol}, we provide a motivation for the solution concept that is adopted for system class~\eqref{eq:sysLin}.
Then, to analyze such solutions, certain properties of the matrix pairs $(E_p, A_p)$, $p \in \N$, are described in Section~\ref{sec:basicProp}.

\subsection{Solution Concept for Switched DAEs}\label{sec:sol}

To discuss the solution concept for the switched DAEs described in \eqref{eq:sysLin}, let us first consider a matrix pair $(E,A)$ and the problem of finding a trajectory $x$ which solves the following initial-trajectory problem (ITP):
\begin{subequations}\label{eq:ITP}
  \begin{align}
    x_{(-\infty,0)}&=x^0_{(-\infty,0)}\label{eq:ITPa}\\
    (E\dot{x})_{[0,\infty)} &= (Ax)_{[0,\infty)}, \label{eq:ITPb}
  \end{align}
\end{subequations}
in some appropriate sense.
In this equation, $x^0$ is some initial trajectory, and $f_\cI$ denotes the restriction of $f$ to an interval $\cI$ and will be defined formally later.
To avoid certain complications (such as nonuniquess of solutions), we limit our attention to \emph{regular} matrix pairs $(E,A)$, i.e.\ we assume that $\det(sE-A)$ is not the zero polynomial.
A fundamental tool in studying \eqref{eq:ITP} is the notion of {\em consistency space}, denoted $\fC_{(E,A)}$, and defined as
\[
   \fC_{(E,A)}:=\left\{x_0\in\R^n \, \vert \, \exists \, x \in\cC^1: E\dot{x}=Ax \wedge x(0)=x_0 \right \}
\]
where $\cC^1$ is the space of differentiable functions $x:\R\to\R^n$.
Intuitively speaking, $\fC_{(E,A)}$ corresponds to the set of initial conditions which are consistent with the algebraic equations encoded in \eqref{eq:ITPb}, and thus a smooth solution evolving within $\fC_{(E,A)}$ can be obtained.
An algebraic characterization of $\fC_{(E,A)}$ in terms of the matrices $E,A$ will be given later in this section.
At this point, it follows from the definition that if $x^0(\cdot)$ is absolutely continuous and $x^0(0^-) \in \fC_{(E,A)}$, then an absolutely continous trajectory $x$ can be defined, such that \eqref{eq:ITP} is satisfied, and furthermore $x(t) \in \fC_{(E,A)}$, for each $t \ge 0$. It is also shown in \citep{Tren09d} that such a solution is unique if, and only if, the matrix pair $(E,A)$ is regular.

However, if $x^0(0^-) \not\in \fC_{(E,A)}$, then it is not so obvious how the solution to \eqref{eq:ITP} must be obtained.
A valid solution $x$ with $x(0^-) = x^0(0^-)$ must ``jump'' to $x(0^+) \in \fC_{(E,A)}$, and from there onwards, the solution can be extended uniquely over the interval $(0,\infty)$.
In contrast to systems with ODEs which only comprise operations involving integration, the DAE \eqref{eq:ITPb} allows the possibility of (higher-order) differentiation.
For this reason, if a jump is introduced at time instant $t = 0$, then the notion of derivatives of jumps must also be introduced in the generalized solution concept associated with \eqref{eq:ITP}.
This motivates us to enlarge the solution space from functions to {\em distributions}, because the derivative of a jump is then formally defined as the {\em Dirac impulse} (or Dirac delta).
But the restriction to an interval as required in \eqref{eq:ITP} cannot be defined for general distributions \citep[Lem.~5.1]{Tren13a}, and hence we consider the smaller space of {\em piecewise smooth distributions}, denoted by $\Dpwsm$, for which \eqref{eq:ITP} is indeed well defined.
We refer the reader to \citep{Tren09d} for formal details, but for this paper, it suffices to recall that $x \in (\Dpwsm)^n$ is written as
\begin{subequations}\label{eq:xGenDist}
\begin{align}
x &= x^f_\D + x[\cdot],
\intertext{where $x_\D^{f}$ denotes the distribution induced by the piecewise smooth function $x^f:\R\to\R^n$ and $x[\cdot]$ denotes the impulsive part of $x$ given by}
   x[\cdot] &= \sum_{k\in\Z} x[t_k] = \sum_{k \in \Z} \sum_{i=0}^{n_k} a_k^i \delta_{t_k}^{(i)},
\end{align}
\end{subequations}
where $\setdef{t_k\in\R}{k\in\Z}$ is a strictly increasing set without finite accumulation and $\delta_{t_k}^{(i)}$ denotes the $i$-th derivative of the Dirac impulse with support at $t_k$. The $n_k+1$ coefficients $a_k^0, a_k^1,\ldots, a_k^{n_k}\in\R^n$ parameterize the \emph{impulsive part} of $x$ at $t_k$, denoted by $x[t_k]$. Furthermore, the \emph{left-} and \emph{right-sided evaluations} $x(t^-)$, $x(t^+)$ are well defined for each $t\in\R$. Finally, a \emph{distributional restriction} is well defined; in fact for $x$ as in \eqref{eq:xGenDist} the restriction to some interval $\cI\subseteq\R$ is defined via
\begin{equation}\label{eq:distr_restr}
   x_{\cI} := (x^f_\cI)_\D + \sum_{t_k\in\cI} x[t_k],
\end{equation}
where $x^f_\cI(t)=x^f(t)$ for $t\in\cI$ and $x^f_\cI(t)=0$ otherwise.

If in \eqref{eq:ITP} it is assumed that $x^0 \in (\Dpwsm)^n$, then we seek a solution $x$ of the form \eqref{eq:xGenDist}, and in this particular case, the impulse times $t_k$ either correspond to the impulse times of $x^0$, or the Dirac impulse that occurs at $t = 0$ due to the jump from an inconsistent initial condition to the consistency space.

When studying switched DAEs of the form \eqref{eq:sysLin}, the consistency spaces for individual subsystems, in general, are different and affected by the presence of inputs as well.
Hence, the switch from one subsystem to another may introduce jumps, and its higher order derivatives which we represent formally using distributions. 
One can draw an analogy between \eqref{eq:ITP} and how the solution of a switched DAE~\eqref{eq:sysLin} is extended beyond a switching time instant: The initial trajectory $x^0 \in \Dpwsm$ corresponds to some past over the interval $(t_0,t_1)$, and then at $t_1$, the system first switches to a new subsystem described by a DAE, which would introduce a jump to the active consistency space, along with the possibility of Dirac impulses.
In general, the state $x(t_p^-)$ may not be consistent with the consistency space of the subsystem $(E_p,A_p,B_p,C_p)$ which is active over the interval $[t_p, t_{p+1})$, and the solution is extended by introducing a jump to the consistency space of the new subsystem, and introducing Dirac impulses at the switching times.
The foregoing discussion motivates us to adopt the framework of piecewise smooth distributions to express solutions of a switched DAE.

\begin{Definition}
The triplet $(x,u,y)$ is called a \emph{(distributional) solution} of the switched DAE \eqref{eq:sysLin} on some interval $\cI\subseteq (t_0,\infty)$ if, and only if, $x \in (\Dpwsm)^n$, $u\in (\Dpwsm)^{\du}$, $y \in (\Dpwsm)^{\dy}$, and  \eqref{eq:sysLin} is satisfied within $\Dpwsm$ on $\cI$, i.e.\
\[
   (E_\sigma\dot{x})_{\cI} = (A_\sigma x+B_\sigma u)_{\cI},\quad y_{\cI} = (C_\sigma x+D_\sigma u)_{\cI}.
\]
If $\cI = (t_0,\infty)$ we call $(x,u,y)$ a \emph{global} solution.
\end{Definition}

In \citep{Tren09d} it is shown that \eqref{eq:sysLin} always has a distributional solution which is uniquely determined on $(t_0,\infty)$ by $x(t_0^+)$ and the input $u$, provided \emph{all matrix pairs $(E_k,A_k)$ are regular}. Throughout this paper we always make this assumption. 

\begin{Remark}\label{rem:local_sol}
Under the regularity assumption and for a fixed input $u$, any solution  of \eqref{eq:sysLin} on some interval $(s,t)$ can be uniquely extended to a solution on $(s,\infty)$; however, it is not always possible to extend this solution in the past as the following simple example shows:
\[
    0 = x \text{ on }(t_0,t_1),\quad\text{and}\quad \dot{x} = 0 \text{ on }[t_1,\infty).
\]
The only global solution is $x=0$, but $x=c$ for any constant $c\in\R^n$ is a solution on $[t_1,\infty)$ which for $c\neq 0$ cannot be extended on the whole interval $(t_0,\infty)$.
\end{Remark}
\subsection{Properties of a matrix pair $(E,A)$} \label{sec:basicProp}

We now collect some important properties and definitions for a regular matrix pair $(E,A)$, which can then of course be applied to each subsystem in \eqref{eq:sysLin}. The properties, and the notations introduced in the process, are then used throughout the paper.

A very useful characterization of regularity is the following well-known result (see e.g.\ \citep{BergIlch12a}).
\begin{Proposition}[Regularity and quasi-Weierstra\ss\  form]
A matrix pair $(E,A)\in\R^{n\times n}\times\R^{n\times n}$ is regular if, and only if, there exist invertible matrices $S,T\in\R^{n\times n}$ such that 
\begin{equation}\label{eq:QWF}
(SET,SAT) = \left( \begin{bmatrix} I & 0  \\  0  &  N \end{bmatrix}, \begin{bmatrix} J  & 0 \\  0  & I \end{bmatrix} \right) ,
\end{equation}
where $J \in \R^{n_1\times n_1}$, $0\leq n_1\leq n$,  is some matrix and $N\in\R^{n_2 \times n_2} $, $n_2:=n-n_1$, is a nilpotent matrix.
\end{Proposition}
The most useful aspect of describing a DAE in quasi-Weierstra\ss\ form \eqref{eq:QWF} is that it decomposes the differential and algebraic part.
If we partition the state $x$ of the system $E\dot x = A x$ after applying coordinate transformation as $(v^\top,w^\top)^\top = T x$ then \eqref{eq:QWF} reads as follows: $\dot v = J v$ and $N\dot w = w$.
The smooth part of the solution is obtained by solving the ODE $\dot v = J v$, and the only classical solution for the equation $N\dot w = w$ is $w=0$. For the latter, the distributional response to an inconsistent initial value is given by
\[
   w[0] = -\sum_{i=0}^{n_2-2}N^{i+1} w(0^-) \delta_0^{(i)}.
\]
We thus obtain an alternate meaning to the solution of a DAE, that is, the ODE component continues to evolve smoothly, whereas the pure algebraic component jumps to the origin with the possibility of Dirac impulses at the time the system is switched on.

To describe the system in the form \eqref{eq:QWF}, one can calculate the matrices $S,T$ by constructing the so called Wong-sequences from the matrices $E,A$, see \citep{BergIlch12a} for details.
Based on these transformations, we define now the following important matrices.

\begin{Definition}\label{def:proj}
 Consider the regular matrix pair $(E,A)$ with corresponding quasi-Weierstra\ss\ form \eqref{eq:QWF}.
 The \emph{consistency projector} of $(E,A)$ is given by
\[
   \Pi_{(E,A)} = T \begin{bmatrix} I & 0\\ 0 & 0 \end{bmatrix}T^{-1}.
\]
Furthermore, let
\[\begin{aligned}
    A^\diff := T \begin{bmatrix} J & 0 \\ 0 & 0 \end{bmatrix} T^{-1}, \text{ and } \ 
    E^\imp := T \begin{bmatrix} 0 & 0 \\ 0 & N \end{bmatrix} T^{-1}.
    \end{aligned}
\]
Finally, if also an output matrix $C$ is considered let
\[
    C^\diff := C \Pi_{(E,A)}.
\]
\end{Definition}

Note that none of the above matrices depends on the specific choice of the (non-unique) quasi-Weierstra\ss\ form \eqref{eq:QWF}.

To give an intuitive interpretation of the objects introduced in Definition~\ref{def:proj}, it is noted that the definition of the consistency projector yields $\im \Pi_{(E,A)}=\fC_{(E,A)}$, and the mapping $\Pi_{(E,A)}$ defines the jump rule onto the consistency space for inconsistent initial conditions, see Lemma~\ref{lem:consProj}.
After the jump, the state evolves within the consistency space, and the matrices $A^\diff$ and $C^\diff$ are used to describe this evolution process, see Lemma~\ref{lem:diffProj}.
The impulsive part of the solution \eqref{eq:ITP} resulting from differentiation of the jump at time $t=0$, denoted by $x[0]$, is described by $E^\imp$, see Lemma~\ref{lem:impulses}.

\begin{Lemma}[{\citep[Thm.~4.2.8]{Tren09d}}]\label{lem:consProj}
   Consider the ITP \eqref{eq:ITP} with regular matrix pair $(E,A)$ and with arbitrary initial trajectory $x^0\in(\Dpwsm)^n$. There exists a unique solution $x\in(\Dpwsm)^n$ and
   \[
      x(0^+)=\Pi_{(E,A)} x(0^-).
   \]
\end{Lemma}

\begin{Lemma}\label{lem:diffProj}
  For any regular matrix pair $(E,A)$ and output matrix $C$, the following implication holds for all continuously differentiable $(x,y)$:
  \[\left.\begin{aligned}
     E\dot{x} &= Ax,\\
     y &= Cx
     \end{aligned}\right\}\ \Rightarrow\ 
     \left\{\begin{aligned}
        \dot{x} &= A^\diff x,\\
        y &= C^\diff x
        \end{aligned}\right.
  \]
  In particular, any classical solution $x$ of $E\dot{x}=Ax$ satisfies
  \[
     x(t) = \me^{A^\diff (t-s)} x(s),\quad s,t\in\R.
  \]
\end{Lemma}
\begin{proof}
The proof is a simple consequence of \citep[Lem.~3]{TanwTren10}.
\end{proof}
 
 \begin{Remark}
    One may wonder why we replace the output matrix $C$ by $C^\diff$. As a motivation, consider a simple DAE with $E = \begin{smallbmatrix} 1 & 0 \\ 0 & 0\end{smallbmatrix}$, $A = I_{2\times 2}$, and $C = [0\ 1]$. For this system, we have $A^\diff = \begin{smallbmatrix} 1 & 0\\0 & 0\end{smallbmatrix}$, and $C^\diff = [0 \ 0]$.
    Hence, with $C^\diff$, it is obvious that we cannot deduce any information about the state from the continuous flow and its corresponding output $y$. However, when using the original output equation the corresponding unobservable space of the ODE is not the whole space, i.e., some parts of the state can be recovered from the output. Although it is true that we can recover $x_2$ from the output, it is misleading because in the DAE, $x_2$ is always zero due to the algebraic constraint. Indeed, the ability to recover $x_2$ does not depend on the actual output matrix $C$.
 \end{Remark}
\begin{Lemma}[{\citep[Cor.~5]{TanwTren10}}]\label{lem:impulses}
Consider the ITP \eqref{eq:ITP} with regular matrix pair $(E,A)$ and the corresponding $E^\imp$ matrix from Definition~\ref{def:proj}.
For the unique solution $x\in(\Dpwsm)^n$ of \eqref{eq:ITP}, it holds that
  \begin{equation}\label{eq:impSol}
     x[0] = -\sum_{i=0}^{n-2}(E^\imp)^{i+1} x(0^-) \delta_0^{(i)},
  \end{equation}
  where $\delta_0^{(i)}$ denotes the $i$-th (distributional) derivative of the Dirac-impulse $\delta_0$ at $t=0$.
\end{Lemma}

\section{Determinability Conditions}\label{sec:obsCond}

To address the problem of state estimation, we first need to study the appropriate notion of observability for switched DAEs.
Several different notions of observability can be defined in the context of switched systems as proposed in a survey chapter \citep{PetrTanw15}.
For the purpose of this paper, it suffices to recall that our approach towards studying observability has three distinctive features compared to the linear time-invariant ODE systems.

\emph{Final-state observability:}
For nonswitched dynamical systems described by linear ODEs, observability relates to recovering the unknown initial value of the state using the input and output measurements.
For switched systems, due to the presence of non-invertible state reset maps (which is the case for switched DAEs), the backward flow in time may not be uniquely defined. Because of this reason, recovering the current state does not imply recovering the initial state (the converse however always holds).
For asymptotic state estimation, it then only makes sense to talk about observability of the current-state, or final-state value at the end of an interval.
This concept is closely related to the notion of final-state observability defined in \citep[Chapter 6]{Sont98a}.

\emph{Large-time observability:}
Typically, for continuous LTI systems, the interval over which the inputs and outputs are measured to recover the state can be arbitrarily small. However, for switched systems, the observation over larger intervals and the change in dynamics due to switching allow us to extract additional information about the state.
This motivates us to study the problem of observability without requiring observability of the individual subsystems in the classical framework, and devise a framework to take different information from multiple modes to derive relaxed conditions.
The concept of large-time observability has also been found useful in the study of nonlinear ODEs \citep{HespLibe05}.

\emph{Algebraic structure:}
The state of a system comprising switched DAEs is constrained by certain algebraic equations in the model description. This is taken into account by the notions of R-observability or behavioral observability for unswitched DAEs (c.f.\ the recent survey \cite{BergReis16a}).
For the model-based state estimation, which is being studied in this paper, computing consistency spaces for individual subsystems from the algebraic equations provides additional information about the state.
Moreover, one may observe the Dirac impulses in the output, which also result due to different algebraic equations allowed in the system class \eqref{eq:sysLin}, and these impulsive measurements prove useful in state estimation. This is related to the notion of impulse observability or observability at infinity (c.f.\ \cite{BergReis16a}), but here we only extract some partial information without assuming that the subsystem are impulse observable.

In the light of the above discussion, we now introduce the notion of {\em determinability}, that combines the properties of large-time and final-state observability. It is formally defined as follows:
\begin{Definition}[Determinability] \label{defn:detObs}
The switched DAE \eqref{eq:sysLin} is called $(s,t]$-\emph{determinable} for $t_0\leq s < t$ if for every pair of local solutions $(x,u,y), (\bar{x},\bar{u},\bar{y})$ of \eqref{eq:sysLin} on $(s,\infty)$ with $u\equiv \bar{u},$ the following implication holds
\begin{equation}\label{eq:det_def}
    y_{(s,t]}=\bar{y}_{(s,t]}\ \Longrightarrow\ x_{(t,\infty)} = \bar{x}_{(t,\infty)}.
\end{equation}
\end{Definition}

Hence, for a fixed switching signal $\sigma$ and given input, the system \eqref{eq:sysLin} is considered determinable over an interval if the output measurements over that interval allow us to reconstruct the value of the state in the future.
Note that the left expression of the implication \eqref{eq:det_def} is equivalent to the conditions $y_{(s,t)}=\bar{y}_{(s,t)}$ and $y[t] = \bar{y}[t]$; in particular the knowledge of the distributional output $y[t]$ will play a major role in the forthcoming observer construction. Furthermore, the right expression of the implication \eqref{eq:det_def} is, due to the regularity assumption, equivalent to the condition $x(t^+) = \bar{x}(t^+)$.


\begin{Remark}
In Definition~\ref{defn:detObs}, we only consider solutions that satisfy \eqref{eq:sysLin} over the interval $(s,\infty)$ and such solutions are not required to satisfy the DAE prior to time $s$.
The reason being, to characterize $(s,t]$-determinability, we are interested in computing the set of indistinguishable states at time $t^+$ using \emph{only} the information from the input and output measurements over the interval $(s,t]$.
To obtain such a characterization it is essential to allow for $x(s^-)$ to be arbitrary;
Otherwise, if $x(s^-)$ is constrained as a solution of a DAE, then that information could possibly reduce the set of states reachable at time $t^+$, see Remark~\ref{rem:local_sol}.
In other words, the constrained past would influence the computation of the set of indistinguishable states which is not desired for our purposes.
\end{Remark}

Based on \citep[Propostion~9]{TanwTren12}, we can restrict our attention to zero-determinability.
Although there is a slight difference in the definition used in this paper, and the one given in \citep{TanwTren12}, the proof of the following result goes through with the same arguments, and is thus not repeated here.
\begin{Proposition}[{\citep[Propostion~9]{TanwTren12}}]\label{prop:fwdObsZero}
The switched DAE \eqref{eq:sysLin} is $(s,t]$-determinable if, and only if, the following implication holds for any local solution $(x,u,y)$ of \eqref{eq:sysLin} on $(s,\infty)$ with $u\equiv 0$:
\[
   y_{(s,t]} \equiv 0 \quad \Longrightarrow \quad x_{(t,\infty)} = 0.
\] 
\end{Proposition}

Thus, in studying whether the state can be completely determined from the knowledge of external measurements, one can ignore the role of inputs and set them to zero.
In essence, for studying determinability conditions, we can reduce our attention to the following system class:
\begin{equation}\label{eq:sysHomLin}
\begin{aligned}
E_\sigma \dot e &= A_\sigma e \\
y^e & = C_\sigma e.
\end{aligned}
\end{equation}
The reason for changing the notation for state variable from $x$ to $e$ in the homogenous system \eqref{eq:sysHomLin} is that later, in the observer design, we will be applying our notion of determinability, and the arguments that follow, to a certain homogenous system of error dynamics. Hence, it is helpful to work with this notation at this point.

The objective now is to compute the set of indistinguishable states for \eqref{eq:sysHomLin} in the sense of Definition~\ref{defn:detObs}.
To do so, we first study the single switch case where we show how the structural decomposition in quasi-Weierstra\ss\  form \eqref{eq:QWF} and the result of Lemmas~\ref{lem:consProj}, \ref{lem:diffProj} and \ref{lem:impulses} allow us to extract the information about the state from the output measured on one switching interval.
When studying the case of multiple switches in Section~\ref{sec:multObs}, we describe how this information is then accumulated at a single time instant to arrive at a characterization of determinability.

\subsection{Observable Information from Single Switch}

Consider the homogeneous switched DAE \eqref{eq:sysHomLin} on $(t_0,\infty)$ and let $t_1>t_0$ be the first switching instant after $t_0$. Let the active subsystem over the interval $(t_0,t_1)$ be denoted by $(E_0,A_0,C_0)$ and assume that the active mode after that is $(E_1,A_1,C_1)$.
We are interested in knowing what information can be deduced about $e(t_1^+)$ using the knowledge of the output $y^e$ of \eqref{eq:sysHomLin} measured on $(t_0,t_1]$ and  the system matrices. Invoking Lemma~\ref{lem:consProj} we know that $e(t_1^+)=\Pi_1 e(t_1^-)$ and it suffices therefore to focus on $e(t_1^-)$ in the following.

\emph{Consistency space:} Independently of the observed output it holds that $e(t_1^-)\in\fC_0$, where $\fC_0:=\fC_{(E_0,A_0)}$ denotes the consistency space of the DAE corresponding to the matrix pair $(E_0,A_0)$.

\emph{Differential unobservable space:} If $y^e_{(t_0,t_1)}\equiv 0$ then $(y^e)^{(i)}(t_1^-)=0$ for all $i\in\N$, hence, invoking Lemma~\ref{lem:diffProj}, we have $e(t_1^-)\in\ker O_0^\diff$, where $\ker O_0^\diff$ denotes the unobservable space of the ODE $\dot{e}=A_0^\diff e$, $y^e=C_0^\diff e$, i.e.\footnote{By $[M_1/M_2/\dots/M_k]$, we denote the matrix which is obtained by stacking the matrices $M_1,M_2,\dots,M_k$ (with the same number of columns)  over each other.}
\begin{equation}\label{eq:observabilitymatrix}
  O_0^\diff:=[C_0^\diff / C_0^\diff A_0^\diff / \cdots / C_0^\diff (A_0^\diff)^{n-1}].
\end{equation}

\emph{Impulse unobservable space:} Finally, due to Lemma~\ref{lem:impulses}, from the equality $0=y^e[t_1]=C_1 e[t_1]$ it follows that $e(t_1^-)\in\ker O_1^{\imp}$, where
  \[
     O_1^\imp:=[C_1 E_1^\imp / C_1 (E_1^\imp)^2 / \cdots / C_1 (E_1^\imp)^{n-1}].
  \]

Altogether this provides the motivation to introduce the \emph{locally unobservable subspace}:
\begin{equation}\label{eq:W1}
\cW_1 := \fC_0 \cap \ker O_0^\diff \cap \ker O_1^\imp.
\end{equation}

\begin{Proposition}\label{prop:fwdObsSing}
The following equality of sets holds:
\[
\cW_1= \setdef{e(t_1^-)}{(e,y^e) \text{ solves \eqref{eq:sysHomLin} and }y^e_{(t_0,t_1]} = 0}.
\]
\end{Proposition}
The proof appears in \ref{app:proofs}.
As a consequence of this result, we have
\[
   \Pi_1 \cW_1 = \setdef{e(t_1^+)}{(e,y^e) \text{ solves \eqref{eq:sysHomLin} and }y^e_{(t_0,t_1]} = 0}
\]
and we call $\Pi_1 \cW_1$ the locally undeterminable space.

\subsection{Gathering Information over Multiple Switches}\label{sec:multObs}

We are now interested in combining information from different modes to characterize how much knowledge about the state can be recovered if we observe measurements over intervals involving multiple switches.
The discussion in the previous section applied to the interval $(t_{p-1},t_p]$, yields that the measurements over an interval $(t_{p-1},t_p]$, $ p \ge 1$, allow us to deduce $e(t_p^-) \in \cW_p$, where
\[
   \cW_p := \fC_{p-1}\cap\ker O_{p-1}^\diff \cap \ker O^{\imp}_{p}
\]
is the locally unobservable space at $t_p$. The objective is to gather the information from previous locally unobservable subspaces and deduce more knowledge about the value of the state at a given time. To do so, we introduce the following sequence of subspaces:
\begin{equation}\label{eq:DAESeqDet}
\begin{aligned}
\cQ_q^{q+1} &:= \cW_{q+1},\\
\cQ_q^{p} & := \cW_{p} \cap \me^{A^\diff_{p-1} \tau_{p-1}} \Pi_{p-1} \cQ_q^{p-1}, \ p > q+1.\\
\end{aligned}
\end{equation}

The intuition behind this definition is as follows: Our aim is to let $\cQ^p_q$ satisfy
\begin{equation}\label{eq:impFwdObsMult}
 \boxed{\cQ^p_q = \setdef{e(t_p^-)}{(e,y_e)=(e,0)\text{ solves \eqref{eq:sysHomLin} on }(t_q,t_p]},}
\end{equation}
in particular, $\Pi_p \cQ_q^p$ is the \emph{$(t_q,t_p]$-undeterminability space} of \eqref{eq:sysHomLin} in the sense that it is the subspace of points $e(t_p^+)=\Pi_p e(t_p^-)$ which cannot be determined from the knowledge of $y^e$ on $(t_p,t_q]$.
Indeed we have the following result.

\begin{Theorem}[Determinability characterization]\label{thm:DAEDetMult}
Consider the homogeneous switched DAE \eqref{eq:sysHomLin} with corresponding consistency projectors $\Pi_p$ and subspaces $\cQ_q^p$, $p>q$, recursively defined by \eqref{eq:DAESeqDet}. Then \eqref{eq:impFwdObsMult} holds for any $p>q\geq 0$. In particular, the switched DAE \eqref{eq:sysLin} is $(t_q,t_p]$-determinable, if and only if,
\begin{equation}\label{eq:CharacDetDAE}
  \boxed{\Pi_p \cQ_q^{p} = \{0 \}.}
\end{equation}
\end{Theorem}
The proof is given in \ref{app:proofs}. Note that the determinability characterization \eqref{eq:CharacDetDAE} is significantly weaker than the condition $\cQ_q^p = \{0\}$ used in \citep{TanwTren13}. In fact, in the latter we aimed to recover $e(t_p^-)$, while for determinability it suffices to determine $e(t_p^+)=\Pi_p e(t_p^-)$ and any uncertainty within  $\ker \Pi_p$ is irrelevant for determinability in the sense of Definition~\ref{defn:detObs}.

\subsection{An illustrative example}\label{sec:example}
To illustrate the above theoretical results, consider an academic example given by \eqref{eq:sysLin} with a periodic switching signal $\sigma$. The first switching time is $t_0 = 0$ and, for $k\in\N$,
\[
    \tau_{4k} = 1,\ \tau_{4k+1} = \pi/4,\ \tau_{4k+2} = 1,\ \tau_{4k+3} = 1.
\]
For $p=0,1,2,3$, the subsystem which is active on the interval $[t_{4k+p},t_{4k+p+1})$ is described by $(E_{4k+p},A_{4k+p},B_{4k+p},C_{4k+p})$, $k \in \N$, defined as,
\[\begin{aligned}
  p=0&:\ \left(I,\begin{smallbmatrix}  -1 & 1 & 0 & 0\\ 0 & 0 & 0 & 0\\ 0 & 0 & 0 & 0\\ 1 & 0 & 0 & 0 \end{smallbmatrix}, \begin{smallbmatrix} 0\\ 1\\ 0\\ 0 \end{smallbmatrix}, \begin{smallbmatrix} 1 & 0 & 0 & 0 \end{smallbmatrix} \right),\\
  p=1&:\ \left(I,\begin{smallbmatrix}  0 & 0 & 0 & 0\\ 0 & 0 & 1 & 0\\ 0 & -1 & 0 & 0\\ -1 & -1 & -1 & -1 \end{smallbmatrix}, \begin{smallbmatrix} 0\\ 0\\ 1\\ 0  \end{smallbmatrix}, \begin{smallbmatrix} 0 & 0 & 0 & 0 \end{smallbmatrix} \right),\\
  p=2&:\ \left(\begin{smallbmatrix} 1 & 0 & 0 & 0\\ 0 & 1 & 0 & 0\\ 0 & 0 & 0 & 0\\ 0 & 0 & 1 & 0 \end{smallbmatrix},\begin{smallbmatrix}  0 & 0 & 0 & 0\\ 0 & 0 & 0 & 0\\ 0 & 0 & 1 & 0\\ 0 & 0 & 0 & 1  \end{smallbmatrix}, \begin{smallbmatrix} 0\\ 0\\ 0\\ 1  \end{smallbmatrix}, \begin{smallbmatrix} 0 & 0 & 0 & 1 \end{smallbmatrix} \right),\\
  p=3&:\ \left(I,\begin{smallbmatrix}  0 & 0 & 0 & 0\\ 1 & -1 & 0 & 0\\ 1 & 0 & 0 & 0\\ 1 & 0 & 0 & 0 \end{smallbmatrix}, \begin{smallbmatrix} 1\\ 0\\ 0\\ 0 \end{smallbmatrix}, \begin{smallbmatrix} 0 & 1 & 0 & 0 \end{smallbmatrix} \right).\\
  \end{aligned}
\]
The subsystems indexed by $p=0,1,3$ are actually ODEs and the one given by $p=2$ is already in quasi-Weierstrass form \eqref{eq:QWF}. Hence, $T_p=S_p=I$, for each $p$. The matrices appearing in Definition~\ref{def:proj}, for $p=0,1,3$, are
\[
   \Pi_p = I,\ A^\diff_p = A,\ C^\diff_p = C,\ E^\imp_p = 0,
\]
and, for $p=2$, are
\[
   \Pi_2 = \begin{smallbmatrix} 1 & 0 & 0 & 0\\ 0 & 1 & 0 & 0\\ 0 & 0 & 0 & 0\\ 0 & 0 & 0 & 0 \end{smallbmatrix},\ A^\diff_2 = 0,\ C^\diff_2 = 0,\  E^\imp_2 = E.
\]
The corresponding locally unobservable spaces are
\[\begin{aligned}
   \cW_1 &= \im \begin{smallbmatrix} 0 & 0\\ 0 & 0\\ 1 & 0\\ 0 & 1 \end{smallbmatrix},&
   \cW_2 &= \im \begin{smallbmatrix} 1 & 0 & 0\\ 0 & 1 & 0\\ 0 & 0 & 0\\ 0 & 0 & 1 \end{smallbmatrix},\\
   \cW_3 &= \im \begin{smallbmatrix} 1 & 0\\ 0 & 1\\ 0 & 0\\ 0 & 0 \end{smallbmatrix},&
   \cW_4 &= \im \begin{smallbmatrix} 0 & 0\\ 0 & 0\\ 1 & 0\\ 0 & 1 \end{smallbmatrix},
  \end{aligned}
\]
where $\im M$ denotes the image of the linear map induced by the matrix $M$ (i.e., its column space).
The $(t_q,t_p]$-undeterminable spaces $\cQ_q^p$ with $\Pi_p Q_q^p = 0$ are as follows
\[
   Q_0^2 = \im \begin{smallbmatrix} 0 \\ 0 \\ 0 \\ 1 \end{smallbmatrix}, \quad Q_2^4 =  \{0\}.
\]
In particular, the switched DAE is determinable on each of the intervals $(t_{4k},t_{4k+2}]$, $(t_{4k+2},t_{4k+4}]$, $k\in\N$.
%

\section{Reconstructing the State}\label{sec:outMaps}

Based on the determinability notions studied in the previous section, we are now interested in knowing how the state of the homogenous error system \eqref{eq:sysHomLin} can actually be reconstructed from the information obtained by (nonzero) output error measurements.
We approach the problem of state reconstruction in the same manner as we derived the determinability conditions in the previous section.
This approach involves two major steps:
\begin{enumerate}
\item Reconstruct the observable components at each switching time from the local (in time) information.
\item Combine the information from each subsystem to find the whole state.
\end{enumerate}

Once we have derived a method to theoretically compute the state of the error dynamics exactly, we can then take additional practical elements, such as estimation errors, into account.

\subsection{Observable Component at a Switching Instant}\label{sec:mapOne}

The first step in reconstructing the state of the error dynamics is to be able to write down a systematic way of extracting the observable information around each switching time from locally active subsystems.
Based on the discussion in the previous section, we know that for any solution $(e,y^e)$ of \eqref{eq:sysHomLin} with $y^e_{(t_0,t_1)} = 0$ and $y^e[t_1] = 0$, the state $e(t_1^-)$ is contained in the locally unobservable space $\cW_1$ given by \eqref{eq:W1}. For a general solution $(e,y^e)$ with \emph{nonzero} $y^e$ the corresponding state vector $e(t_1^-)$  is not contained in $\cW_1$.

\emph{In the following, we will decompose $e(t_1^-)$ according to the direct sum $\R^n=\cW_1\oplus\cW_1^\bot$.} For that, let us introduce  orthonormal matrices $W_1$ and $Z_1$ such that $\im W_1 = \cW_1$ and $\im Z_1 = \cW_1^\bot$. Let $d_{w_1}$ denote the dimension of $\cW_1$, then there exist unique elements $w_1 \in \R^{d_{w_1}}$, and $z_1 \in \R^{n-d_{w_1}}$ such that
\begin{equation}\label{eq:xZW}
e(t_1^-) = W_1 w_1 + Z_1 z_1.
\end{equation}
Because of orthonormality of the matrix $Z_1$, multiplying the above equation from left by $Z_1^\top$, we get $z_1 = Z_1^\top e(t_1^-)$.
Similarly, multiplication from left by $W_1^\top$ gives $w_1=W_1^\top e(t_1^-)$.
It is noted that $w_1$ is the unobservable component of $e(t_1^-)$ and $z_1$ the component of the vector $e(t_1^-)$ that can be recovered from measuring $y^e_{(t_0,t_1)}$ and $y^e[t_1]$. In the sequel, we construct a mapping that relates the vector $z_1$ to the measurements of $y_{(t_0,t_1]}^e$.

We first make the observation that
\begin{align*}
\cW_1^\bot & = {(\fC_0 \cap \ker O_0^\diff \cap \ker O_1^{\imp})}^\bot \\
& = \fC_0^\bot + \im\left({O^\diff_0}^\top\right) + \im\left({O_{1}^{\imp}}^\top\right),
\end{align*}
that is, $\cW_1^\bot$ is a sum of three subspaces. Since $z_1$ corresponds to the projection of $e(t_1^-)$ onto the subspace $\cW_1^\bot$, we further decompose the vector $z_1$ along each of the three constituent subspaces of $\cW_1^\bot$, which are subsequently denoted by $z_1^\cons$, $z_1^\diff$, and $z_1^\imp$.
More formally, we let $Z_0^\cons$, $Z_0^\diff $, and $Z_1^\imp $ be matrices with orthonomal columns such that
\[\begin{aligned}
   & \im Z_0^\cons = \fC_0^\bot, & \quad & z_1^\cons := {Z_0^{\cons}}^\top e(t_1^-),\\
   & \im Z_0^\diff = \im\left({O^\diff_0}^\top\right), & \quad & z_1^\diff := {Z_0^{\diff}}^\top e(t_1^-), \\
   & \im Z_1^\imp = \im\left({O^\imp_1}^\top\right), & \quad & z_1^\imp := {Z_1^{\imp}}^\top e(t_1^-).
\end{aligned}
\]
Note that the images of the matrices $Z_0^\cons$, $Z_0^\diff$ and $Z_1^\imp$ might intersect non-trivially with each other. In this case, some part of the unknown error $e(t_1^-)$ can be determined from the consistency or classically differentiable part as well as from the impulsive information. From a mathematical point of view this redundancy can be eliminated by choosing a full column rank matrix $U_1$ such that 
\begin{equation}\label{eq:zbardeff}
[ Z_0^\cons \, , \ Z_0^\diff \, , \ Z_1^\imp ]\, U_1 = Z_1,
\end{equation}
and for brevity we let $\overline{Z}_1:=[ Z_0^\cons \, , \ Z_0^\diff \, , \ Z_1^\imp ]$.
We thus obtain,
\begin{equation}\label{eq:z1Comp}
z_1 = U_1^\top \overline{Z}_1^\top e(t_1^-) = U_1^\top \begin{pmatrix} z_1^\cons / z_1^\diff / z_1^\imp \end{pmatrix}.
\end{equation}
We now describe, how each component of the vector $z_1$ can be calculated.

\emph{The consistency information $z_1^\cons$:} 
By definition, it holds that ${Z_0^\cons}^\top \fC_0 = \{0\}$ and hence $z_1^\cons = {Z_0^{\cons}}^\top e(t_1^-) = 0$ because any solution of the homogenous DAE \eqref{eq:sysHomLin} evolves within the consistency space $\fC_0$ on the interval $(t_0,t_1)$.

\emph{Mapping for the differentiable part $z_1^\diff$:} In order to define $z_1^\diff \in \R^{r_0}$, where $r_0 = \rank O_0^\diff$, we first introduce the function $\mathbf{z}_0^\diff:(t_0,t_1) \rightarrow \R^{r_0}, t\mapsto {Z_0^\diff}^\top e(t)$, which represents the observable component of the subsystem $(E_0, A_0, C_0)$ that can be recovered from the smooth output measurements $y^e$ over the interval $(t_0,t_1)$. It follows from Lemma~\ref{lem:reduced_obs} in \ref{app:lemmas} that the evolution of $\mathbf{z}_0^\diff$ is governed by an observable ODE
\begin{equation}\label{eq:zsys}
\begin{aligned}
 \dot{\mathbf z}_0^\diff &= S_0^\diff \mathbf{z_0}^\diff, \\
  y^e &= R_0^\diff {\mathbf z}_0^\diff, 
\end{aligned}
\end{equation}
where $S_0^\diff:={Z_0^\diff}^\top A_0^\diff Z_0^\diff$ and $R_0^\diff:=C_0^\diff Z_0^\diff$. Because of the observability of the pair $(S_0^\diff, R_0^\diff)$ in \eqref{eq:zsys}, there exists a (linear) map $\cO_{(t_0,t_1)}^\diff$ such that
\[
\mathbf z_0^\diff = \cO_{(t_0,t_1)}^\diff (y^e_{(t_0,t_1)})
\]
and we set
\[
z_1^\diff = \mathbf z_0^\diff(t_1^-).
\]
From the vast literature that exists on observability of classical linear time-invariant systems, one can find various ways for representing and approximating the map $\cO^\diff_{(t_0,t_1)}$. For our purposes, this choice is not essential.
We are just interested in knowing that there exist techniques which allow us to (approximately) recover $z_1^\diff$ and in the design of our estimators in Section~\ref{sec:obsDesign}, it will be specified what (approximation) properties are required.

\emph{Mapping for the impulsive part $z_1^\imp$:}
A particular characteristic of the switched DAEs is that the different algebraic constraints for different modes may lead to Dirac impulses in the solution trajectories at switching times. If such impulses are observed in the output, then this information can be used to recover a certain part of the state $e(t_1^-)$. The impulsive part of the output at switching time $t_1$ can be represented as
\[
   y^e[t_1] = \sum_{i=0}^{n-2} \eta_1^i \delta_{t_1}^{(i)},
\]
where due to Lemma~\ref{lem:impulses} the coefficients $\eta_1^i$ satisfy the relation $\boldsymbol{\eta}_1 = -O_{1}^{\imp} e(t_1^-)$, with $\boldsymbol{\eta}_1 := ({\eta_1^0}/\cdots/{\eta_1^{n-2}})\in\R^{(n-1) \dy}$.
We now want to find a linear map $\cO_{[t_1]}^\imp$ such that
$
z_1^\imp = \cO_{[t_1]}^\imp (y^e[t_1]).
$
For that, we chose a matrix $U_1^\imp$ such that
\begin{equation}\label{eq:U1imp}
   -{O_{1}^{\imp}}^\top U_1^\imp = Z_1^\imp,
\end{equation}
 then
\begin{equation}\label{eq:z1imp}
  z_1^\imp = Z_1^{\imp^\top} e(t_1^-) = - U_1^{\imp^\top} O_{1}^{\imp}  e(t_1^-) = {U_1^\imp}^\top \boldsymbol{\eta}_1, 
\end{equation}
hence $\cO_{[t_1]}^\imp$ is given by:
\[
    y^e[t_1]= \sum_{i=0}^{n-2} \eta_1^i \delta_{t_1}^{(i)}\quad \mapsto\quad {U_1^\imp}^\top (\eta_1^0/\ldots/\eta_1^{n-2}).
\]

\subsection{Construct Global Mapping from Local Mappings}\label{sec:mapComb}

In the previous subsection, we described how the information over an interval with one switch can be combined with the data of active subsystems over that interval to reconstruct the observable information from that data at the switching time.
This can obviously be done at each switching time $t_k$, $k\in\N$, that is, by looking at the intervals $(t_{k-1},t_k]$ and measuring $y^e_{(t_{k-1},t_k)}$ and $y^e[t_k]$, one can repeat the procedure in Section~\ref{sec:mapOne} to compute $z_k$ where
\begin{equation}\label{eq:zpZp}
z_k = Z_k^\top e(t_k^-)
\end{equation}
and $Z_k$ is an orthonormal matrix with range space $\cW_k^\bot$.
Because of our notion of determinability, the next step is to be able to combine these local informations obtained at each switching time to reconstruct the entire state of the system \eqref{eq:sysHomLin}.
More specifically, assuming that \eqref{eq:sysHomLin} is $(t_q,t_p]$-determinable for some $q,p\in \N$, we are interested in finding a linear map $\cO_q^p$ such that
\[
e(t_p^+) = \Pi_p\cO_q^p (\mathbf{z}_{q+1}^p),
\]
where $\mathbf{z}_{q+1}^p := (z_{q+1}/z_{q+2}/\ldots/z_p)$. We next construct this map $\cO_q^p$.
A schematic representation of the development of this section is given in Figure~\ref{fig:combInfo}.

\begin{figure}[hbt]
\centering
\scalebox{0.9}{\begin{tikzpicture}
\draw (0,0) rectangle (7,5);
\draw[fill = red!20] (0,0) rectangle (2,5);
\draw (1,1.5) node [text width = 1.8cm, align = center, inner sep = 0] {\footnotesize Undet.\ states over $(t_q,t_{q+1}]$};
\draw (1,3) node [fill=green!20, rectangle, draw, align=center, text width = 2cm, minimum height = 0.5cm, inner sep = 0, text centered] (z1) {$z_{q+1}$};
\draw[fill = red!20] (2,0) rectangle (4,5);
\draw (3,2) node [text width = 1.8cm, align = center, inner sep = 0] {\footnotesize Undet.\ states over $(t_{q+1},t_{q+2}]$};
\draw (3,4) node [fill=green!20, rectangle, draw, align=center, text width = 2cm, minimum height = 0.5cm, inner sep = 0, text centered] (z2) {$z_{q+2}$};
\draw (4.5,2) node [rectangle, align=center, inner sep = 0, text width = 1cm, text centered] (dots) {$\dots$};
\draw[fill = red!20] (5,0) rectangle (7,5);
\draw (6,2) node [text width = 1.8cm, align = center, inner sep = 0] {\footnotesize Undet.\ states over $(t_{p-1},t_{p}]$};
\draw (6,0.5) node [fill=green!20, rectangle, draw, align=center, text width = 2cm, minimum height = 0.5cm, inner sep = 0, text centered] (zm) {$z_{p}$};
\draw (8.5,2.5) node [rotate=-90, fill= green!20, rectangle, draw, align=center, text width = 5cm, minimum height = 1cm, inner sep = 0, text centered] (xi1) {\small $e(t_p^+) = \Pi_p\cO_q^p (z_{q+1}, z_{q+2}, \dots, z_p)$};
\draw[thick,->] (z1.east) -- +(6,0);
\draw[thick,->] (z2.east) -- +(4,0);
\draw[thick,->] (zm.east) -- +(1,0);
\draw[->] (0,-0.5) -- (9,-0.5) node [anchor=west] {$t$};
\draw (0,-0.4) -- (0,-0.6) node [anchor = north] {$t_q$}; 
\draw (2,-0.4) -- (2,-0.6) node [anchor = north] {$t_{q+1}$}; 
\draw (4,-0.4) -- (4,-0.6) node [anchor = north] {$t_{q+2}$}; 
\draw (5,-0.4) -- (5,-0.6) node [anchor = north] {$t_{p-1}$}; 
\draw (7,-0.4) -- (7,-0.6) node [anchor = north] {$t_{p}$}; 
\end{tikzpicture}}
\caption{Accumulating local information $z_k$, $k=q,q+1,\ldots,p$ at time instant $t_p$ to calculate $e(t_p^+)$.}
\label{fig:combInfo}
\end{figure}
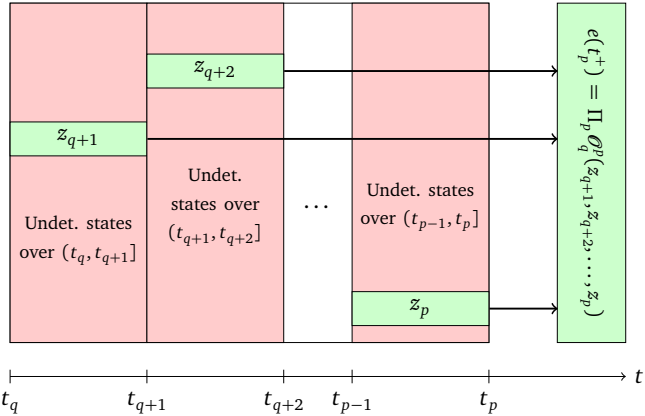

For $k\in\N$ with $q < k \leq p $ let $P_{q}^{k}$ and $Q_{q}^{k}$ be matrices with orthonormal columns such that
\[
   \im Q_q^k = \cQ_{q}^{k}\quad\text{and}\quad \im P_q^k = (\cQ_{q}^{k})^\bot,
\]
where $\cQ_q^k$ is recursively defined by \eqref{eq:DAESeqDet} and is the subspace of all points $e(t_k^-)$ which cannot be determined from $y^e_{(t_q,t_k]}$. There exist unique vectors $\varphi_q^k$, and $\chi_q^k$ such that 
\begin{equation}\label{eq:xPkQk}
   e(t_{k}^-) = P_{q}^{k} \varphi_{q}^{k} + Q_{q}^{k} \chi_{q}^{k}.
\end{equation}
Because of orthonormality of the matrices $P_q^k$ and $Q_q^k$, we have that $\varphi_{q}^{k} = {P_{q}^{k}}^\top e(t_k^-)$ and $\chi_{q}^{k} = {Q_{q}^{k}}^\top e(t_{k}^-)$. In particular, if \eqref{eq:sysHomLin} is $(t_q,t_p]$-determinable, we have $\cQ_q^p \subseteq \ker \Pi_p$ and therefore 
\[
   e(t_p^+) = \Pi_p P^p_q \varphi_q^p,
\]
i.e.\ the state $e(t_p^+)$ can be recovered if we are able to find an expression for $\varphi_q^k$, $k=q+1,q+2,\ldots,p$.
Note that $\varphi_q^{q+1} = z_{q+1}$ corresponds to the determinable information from the interval $(t_q,t_{q+1}]$, and we already discussed in \ref{sec:mapOne} how it can be obtained. We now derive a recursive expression for $\varphi_q^k$, $k=q+2,q+3,\ldots,p$ in terms of $\varphi_q^{k-1}$.

For that we need to introduce a matrix $\Theta_q^{k-1}$ with orthonormal columns such that
\begin{equation}\label{eq:Thetaqk}
   \im \Theta_q^{k-1} = \left(\me^{A_{k-1}^\diff\tau_{k-1}}\Pi_{k-1} \cQ_{q}^{k-1}\right)^\bot.
\end{equation}
Note that then by definition
\[
    \im P_q^k = {\cQ_q^k}^\bot = \im [Z_k, \Theta_q^{k-1}],
\]
in particular there is a matrix $U_q^k$ such that
\[
   P_q^k = [Z_k, \Theta_q^{k-1}] U_q^k.
\]
From
\begin{align}
  e (t_k^-)  & = \me^{A_{k-1}^\diff \tau_{k-1}} \Pi_{k-1} \ e(t_{k-1}^-)\notag \\
    & = \me^{A_{k-1}^\diff \tau_{k-1}} \Pi_{k-1} \left( P_{q}^{k-1} \varphi_{q}^{k-1} +  Q_{q}^{k-1} \chi_{q}^{k-1} \right), \label{eq:err2b}
\end{align}
together with
\[
   {\Theta_{q}^{k-1}}^\top \me^{A_{k-1}^\diff\tau_{k-1}}\Pi_{k-1} Q_{q}^{k-1} = 0.
\]
we can conclude that
\begin{align}
   \varphi_q^k &= {P_q^k}^\top e(t_k^-) =  {U_q^k}^\top \begin{bmatrix} Z_k^\top \\ {\Theta_q^{k-1}}^\top \end{bmatrix} e(t_k^-)\notag\\
   &=   {U_q^k}^\top \begin{pmatrix} z_k \\ {\Theta_q^{k-1}}^\top \me^{A_{k-1}^\diff \tau_{k-1}} \Pi_{k-1} P_{q}^{k-1} \varphi_{q}^{k-1} \end{pmatrix},\label{eq:phiqk}
\end{align}
which is the desired recursion formula for $\varphi_q^k$ for $k=q+2,q+3,\ldots,p$ with ``initial value'' $\varphi_q^{q+1} = z_{q+1}$.

We have thus arrived at the following result:
\begin{Theorem}\label{thm:mapZs}
Consider the homogenous switched DAE \eqref{eq:sysHomLin} with corresponding $A^\diff_p,\Pi_p$ as in Definition~\ref{def:proj}, $\tau_p:=t_{p+1}-t_p$, $p\in\N$, and assume $(t_q,t_p]$-determinability for some $0\leq q < p$. Furthermore, consider for $k=q+2,q+3,\ldots,p$ the matrices $U_q^k$, $\Theta_q^{k-1}$, $P_q^{k-1}$ as above and let
\[\begin{aligned}
   F_q^k &:= {U_q^k}^\top  \begin{bmatrix} I \\ 0 \end{bmatrix},\quad\text{and}\quad F_q^{q+1}:=I,\\
   G_q^k &:= {U_q^k}^\top  \begin{bmatrix}0 \\ {\Theta_q^{k-1}}^\top \me^{A_{k-1}^\diff \tau_{k-1}} \Pi_{k-1} P_q^{k-1}\end{bmatrix}
   \end{aligned}
\]
and with $\mathbf{z}_{q+1}^p=(z_{q+1},z_{q+2},\cdots,z_p)$ let \footnote{
For $q+1 \le k \le p-1$, the notation $\left(\prod_{j=k+1}^{p}G_q^{j} \right)$ denotes the product of matrices $G_q^p \, G_q^{p-1} \dots \, G_q^{k+1}$ (note the decreasing order).
By convention, when $p \le k$, the resulting product is set to identity.
}
\begin{equation}\label{eq:Oqp}
   \cO_q^p(\mathbf{z}_{q+1}^p) := P_q^p \sum_{k=q+1}^p \left(\prod_{j=k+1}^p G_q^j\right) F_q^k z_k.
\end{equation}
Then
\[
  \boxed{e(t_{p}^+) = \Pi_p\cO_q^{p}(\mathbf{z}_{q+1}^{p})},
\]
i.e., the linear map $\cO_q^p$ describes how the state $e(t_p^+)$ can be recovered from the knowledge of locally observable parts $z_{q+1}$, $z_{q+2}$, \ldots, $z_p$, for which the construction was provided in Section~\ref{sec:mapOne}.
\end{Theorem}

\begin{Remark}[Dependence on switching signal]
  While the reconstruction of the locally observable component $z_k$ at the switching time $t_k$ only depends on the two modes that are active prior and after the switch, i.e., $(E_{k-1},A_{k-1},C_{k-1})$ and $(E_k,A_k,C_k)$, the overall reconstruction of the state, as in Theorem \ref{thm:mapZs}, additionally depends on the duration $\tau_k=t_{k+1}-t_k$ of each mode, because the matrix $\Theta_q^{k-1}$ used in the construction depends on $\tau_{q+1}, \dots, \tau_{k-1}$. Hence the map $\cO_q^p$ depends on $\tau_{q+1}, \tau_{q+2}, \ldots, \tau_p$ (but not on $\tau_q$).
\end{Remark}

\subsection{Example}
We revisit our example from Section~\ref{sec:example} and recall that it is $(t_0,t_2]$ determinable, so that the above derivation can be used to recover the state $e(t_2^+)$ of the homogenous switched DAE \eqref{eq:sysHomLin} from the output $y^e$ on $(t_0,t_2]$. For the reconstruction via \eqref{eq:phiqk}, the following matrices are needed:
\[
   P_0^1 = \begin{smallbmatrix} 1 & 0 \\ 0 & 1 \\ 0 & 0 \\ 0 & 0 \end{smallbmatrix},\
   \Theta_0^1 = \begin{smallbmatrix} 0 & 1 \\ -\sqrt{2}/2 & 0 \\ \sqrt{2}/2 & 0 \\ 0 & 0\end{smallbmatrix},\ 
   U_0^2 = \begin{smallbmatrix} 0 & 1 & 0 \\0 & -\sqrt{2} & 0 \\ 1 & 0 & 0 \end{smallbmatrix},
\]
\[
   \me^{A^\diff_1\tau_1}\Pi_1 \approx \begin{smallbmatrix} 1 & 0 & 0 & 0 \\ 0 & \sqrt{2}/2 & \sqrt{2}/2 & 0 \\ 0 & -\sqrt{2}/2 & \sqrt{2}/2 & 0 \\ -0.544 & -0.251 & -\sqrt{2}/2 & 0.456 \end{smallbmatrix}.
\]
Furthermore, for estimating $z^\diff_1$ via \eqref{eq:zsys}, an observer for the following ODE has to be implemented:
\[
   \dot{\mathbf{z}}_0^\diff = \begin{smallbmatrix} -1 & 0 \\ 0 & 0 \end{smallbmatrix} \mathbf{z}_0^\diff, \quad y^e = \begin{smallbmatrix} 1 & 0 \end{smallbmatrix} \mathbf{z}_0^\diff.
\]
Finally, for estimating $z^\imp_2$ via \eqref{eq:z1imp} from the impulses in the output, we use the compression matrix $U_2^\imp = [1\ 0\ 0\ 0]^\top$.

\section{Design of State Estimators}\label{sec:obsDesign}

In the previous section, we provided a method for reconstructing the state of the homogenous error dynamic system (without inputs) by constructing the mappings that exist between the output error and the observable components of the individual system.
In practice, these mappings are not numerically robust, or physically realizable, and we are thus interested in obtaining estimates of the state trajectories through numerically robust algorithms.
So, the purpose of this section is to design an observer for the system class \eqref{eq:sysLin} under the interval wise determinability assumption introduced in Section~\ref{sec:obsCond}, which generates asymptotically convergent state estimates.
The underlying structure of these estimators is based on the result of Theorem~\ref{thm:mapZs} developed in Section~\ref{sec:outMaps}.

Our proposed observer for the system class~\eqref{eq:sysLin} is given by (see also Figure~\ref{fig:obsAll}):
\begin{subequations}\label{eq:obsSyn}
\begin{empheq}[box=\fbox]{align}
E_{\sigma} \dot{\widehat{x}}_p &= A_{\sigma} \widehat{x}_p + B_{\sigma} u, \text{ on } [t_p,t_{p+1}^+], \label{eq:obsSyna}\\
 \widehat{x}_{p}(t_{p}^-) &= \widehat{x}_{p-1}(t_p^-) - \xi_p \label{eq:obsSynb}, \quad p\in\N,
\end{empheq}
\end{subequations}
with arbitrary initial condition $\widehat{x}_0(t_0^-) \in \R^n$ and $[t_p,t_{p+1}^+]:=[t_p,t_{p+1}+\eps)$ for some arbitrarily small $\eps>0$. The desired estimate $\widehat x$ is defined as:
\[
   \widehat{x}:=\sum_{p \in \N} (\widehat{x}_p)_{[t_{p},t_{p+1})}.
\]

It is seen that the observer consists of a system copy and unlike classical methods where the continuous dynamics of the estimate are driven by an error injection term, {\em the observer \eqref{eq:obsSyn} updates the state estimate only at discrete switching instants by an error correction vector $\xi_p$}, which is determined by the difference between the system output $y=C_\sigma x$ and system copy output $\widehat{y}=C_\sigma \widehat{x}$. To give an intuitive interpretation of how to calculate $\xi_p$, note that, \emph{under the assumption} that for some $p \in \N$, the correction term $\xi_p$ satisfies
\[
   \Pi_p \xi_p = \Pi_p(\widehat{x}_{p-1}(t_p^-) - x(t_p^-)),
\]
with $x$ being a solution of \eqref{eq:sysLin},
the equation~\eqref{eq:obsSynb} gives $\widehat{x}_p(t_p^+) = \Pi_p\widehat{x}_p(t_p^-) = \Pi_p x(t_p^-) = x (t_p^+)$, and from there onwards the system copy \eqref{eq:obsSyn} with $\xi_k=0$, $k>p$ will follow exactly the original system, at least in theory.
 In reality, however, even after a perfect match at time $t_p$, the system copy will deviate again from the original system due to some uncertainties and we have to apply a correction $\xi_{\tilde{p}}$ at a later switching time $t_{\tilde{p}}$ again. It is not necessary (and may also not be possible) to apply this correction at every step. So our goal is to compute $\xi_p$ repeatedly, for ``sufficiently many'' $p \in \N$, such that $\Pi_p\xi_p$ approximates the value of state estimation error at time $t_p^+$ closely as $p$ gets large.
Since the growth of error between the reset times can be upper bounded by the solution of a linear ODE, the resets (under determinability assumption) allow us to make the estimation error at switching times arbitrarily small, which eventually results in convergence of $\widehat {x}(t)$ toward $x(t)$ as $t$ tends to infinity.

With this motivation, we introduce the state estimation error. Let $e_p := \widehat{x}_p - x$ denote the state estimation error on $[t_{p},t_{p+1})$ and $e:=\sum_p (e_p)_{[t_{p},t_{p+1})} = \widehat{x} - x$, then
\begin{subequations}\label{eq:errDyn}
\begin{align}
E_p \dot e_p &= A_p e_p,\quad\text{ on }[t_p,t_{p+1}), \label{eq:errDyna}\\
e_{p}(t_p^-) & = e_{p-1}(t_p^-) - \xi_{p}. \label{eq:errDynb}
\end{align}
\end{subequations}
Note that equations~\eqref{eq:obsSyna} and \eqref{eq:errDyna} are both to be interpreted in the sense of distributions, in particular, the impulsive parts $x_p[t_p]$ and $e_p[t_p]$ are uniquely determined by \eqref{eq:obsSyn} and \eqref{eq:errDyn}, respectively. However, the error dynamics \eqref{eq:errDyna} are homogenous and there are no impulses between two switches. As a result, the solution of \eqref{eq:errDyna} for $t \in (t_{p},t_{p+1})$ is described as
\begin{equation}\label{eq:errFlowOneSwitch}
\begin{aligned}
e(t) = e_p(t) & = \me^{A_p^\diff(t-t_{p})} \Pi_p e_p(t_{p}^-) \\
& = \me^{A_p^\diff(t-t_{p})} \Pi_p \left(e(t_{p}^-) - \xi_p \right).
\end{aligned}
\end{equation}
The output estimation error is
\[
y^e = C_p \widehat{x}_p - y
\]
on each open interval $(t_{p},t_{p+1})$. The impulsive error $y^e[t_p]$ at the switching times is obtained by the difference between $y[t_p]$ and the output impulse resulting from \eqref{eq:obsSyna} \emph{without} taking the correction $\xi_p$ into account, that is,
\begin{equation}\label{eq:impulse_error}
   y^e[t_p] := C_p \widehat{x}_{p-1}[t_p] -y[t_p]. \\
\end{equation}
Note that, in general,
$
y^e[t_p] \neq C_{p} \widehat{x}[t_p] - y[t_p],
$
because $\widehat{x}[t_p] = \widehat{x}_p[t_p] \neq \widehat{x}_{p-1}[t_p]$. In fact, the knowledge of $y^e[t_p]$ is based on the knowledge of $\widehat{x}_{p-1}$ which will be used to determine $\xi_p$, which in turn determines $C_p \widehat{x}_p[t_p]$.
Furthermore, note that $y[t_p]$ as well as $\widehat{x}_{p-1}[t_p]$ depends on $u[t_p]$ and the derivates $u^{(i)}(t_p^+)$ immediately after time $t_p$ (see \citep[Thm.~6.5.1]{Tren12}). This will render the observer slightly acausal, as the information immediately after $t_p$ is used to set the value of $\widehat{x}_p(t_p^-)$.
For DAEs, this is not surprising because the transfer functions are not necessarily proper and hence involve differentiation.
However, this is not a serious problem from an implementation-point-of-view because the system copy \eqref{eq:obsSyn} is not required to run synchronously to the original system.
Another way to overcome this issue is to assume that the input is smooth at the switching instants (i.e.\ jumps in the inhomogeneity are induced by a switching $B$-matrix), then $u^{(i)}(t_p^+)=u^{(i)}(t_p^-)$.

Now that we have described the homogenous (input-free) dynamics for the state estimation error, and the corresponding output equation, we are interested in computing the vector $\xi_p$ such that $\Pi_p \xi_p$ estimates $e_{p-1}(t_p^+) = \Pi_p e_{p-1}(t_p^-)$.
We make use of the analysis carried out in Section~\ref{sec:outMaps} to implement the following basic idea in calculating $\xi_p$:

\emph{Step 1:} Identify the observable component $z_p$ of the individual subsystems for the error dynamics \eqref{eq:errDyn}. For subsystem $p\in \N$, we let $Z_p$ be an orthonormal matrix with range space $\cW_p^\bot$, in particular, $z_p = Z_p^\top e(t_p^-)$. 

\emph{Step 2:} Under the assumption that for $p \in \N$, there exists a positive integer $q$ such that $(t_q, t_p]$-determinability holds, we derive a linear function $\Xi_{q}^{p-1}(\cdot)$ such that any solution $e$ of the error dynamics \eqref{eq:errDyn} satisfies
\begin{equation}\label{eq:preXi}
 \Pi_p e (t_p^-) = \Pi_p \left(\cO_{q}^p (\mathbf{z}_{q+1}^p) - \Xi_{q}^{p-1}(\boldsymbol{\xi}_{q+1}^{p-1})\right)
\end{equation}
where 
\[
   \boldsymbol{\xi}_{q+1}^{p-1} = (\xi_{q+1},\xi_{q+2},\ldots,\xi_{p-1}).
\]
and $\cO_q^p$ is given by \eqref{eq:Oqp}.

\emph{Step 3:} Finally, the estimates $\widehat z_k$ for the observable components $z_k$ are constructed at times $t_k^-$, $q+1 \leq k \le p$, and the error correction vector $\xi_p$ is defined as
\begin{equation}\label{eq:preErrC}
   \xi_p := \cO_{q}^p(\widehat{\mathbf{z}}_{q+1}^p) - \Xi_{q}^{p-1}(\boldsymbol{\xi}_{q+1}^{p-1}), \\
\end{equation}
where
$
   \widehat{\mathbf{z}}_{q+1}^p = (\widehat{z}_{q+1},\widehat{z}_{q+2},\ldots,\widehat{z}_p).
$

The ``error correction'' block is basically a compact representation of the structure given in Figure~\ref{fig:combInfo}, where additionally the effect of the state resets $\xi_p$ are taken into account.
In Section~\ref{sec:outMaps}, we have shown how to compute the map $\cO_{q}^p$ using the output measurements for a homogeneous system. For computing the error correction vector $\xi_p$, we use the same operator but applied to the estimates of the observable components.
Computations of the estimates $\widehat{z}_p$ of components of $z_p$, $p \in \N$, will be discussed in Section~\ref{sec:zEst}.
Before that, we close this section by deriving the second missing component for computing $\xi_p$ in \eqref{eq:preErrC}, which is the map $\Xi_{q}^{p-1}$ and it appears due to the state resets at previous switching instants \eqref{eq:errDynb}\footnote{This analysis can be skipped, if the state-resets are only applied at the end of a determinable interval $(t_q,t_p]$. However, in general we allow a reset of the estimator state at any switching time.}.
From \eqref{eq:preXi}, it is clear that the computation of $\Xi_{q}^{p-1}$ requires us to write $\Pi_p e(t_p^-)$ in terms of $z_{q+1}, \dots,z_{p}$, and $\xi_{q+1}, \dots, \xi_{p-1}$.
Analogously to the derivation leading to Theorem~\ref{thm:mapZs}, there exist unique vectors $\psi_q^k$ and $\chi_q^k$ such that, cf.\ \eqref{eq:xPkQk},
\[
   e(t_{k}^-) = P_{q}^{k} \psi_{q}^{k} + Q_{q}^{k} \chi_{q}^{k},\quad k=q+1,q+2,\ldots,p,
\]
where, as before, $P_q^k$ and $Q_q^k$ are orthonormal matrices whose columns span $(\cQ_q^k)^\bot$ and $\cQ_q^k$, respectively, and $\cQ_q^k$ is given by \eqref{eq:DAESeqDet}. Invoking
\[
   e(t_k^-) = \me^{A_{k-1}^\diff\tau_{k-1}} \Pi_{k-1} \left(e(t_{k-1}^-) - \xi_{k-1}\right), 
\]
we arrive at the following recursion formula, cf.\ \eqref{eq:phiqk},
\[
   \psi_q^k = {U_q^k}^\top \begin{pmatrix} z_k \\ {\Theta_q^{k-1}}^\top \me^{A_{k-1}^\diff \tau_{k-1}} \Pi_{k-1}\left(P_q^{k-1}\psi_q^{k-1} - \xi_{k-1}\right)\end{pmatrix}
\]
with ``initial condition'': $\psi_q^{q+1} = z_{q+1}$. We thus obtain the following generalization of Theorem~\ref{thm:mapZs}:
\begin{Proposition}\label{prop:mapCompRecur}
Consider the error system \eqref{eq:errDyn} and assume $(t_q,t_p]$-determinability of \eqref{eq:sysLin} for some $0\leq q < p$. Using the notation from Theorem~\ref{thm:mapZs} let
\[
    H_q^k := {U_q^k}^\top \begin{bmatrix}  0 \\ {\Theta_q^{k-1}}^\top \me^{A_{k-1}^\diff \tau_{k-1}} \Pi_{k-1} \end{bmatrix}
\]
and
\begin{equation}\label{eq:Xiqp-1}
    \Xi_q^{p-1}(\boldsymbol{\xi}_{q+1}^{p-1}) := \sum_{k=q+1}^{p-1} \left(\prod_{k+2}^{p} H_q^j\right) \xi_{k},
\end{equation}
then
\[
  \Pi_p e(t_p^-) = \Pi_p \left(\cO_q^p(\mathbf{z}_{q+1}^p) - \Xi_q^{p-1}(\boldsymbol{\xi}_{q+1}^{p-1})\right).
\]
\end{Proposition}
%

\subsection{Observable Components and their Estimates}\label{sec:zEst}

Before presenting the main convergence result, the last ingredient for computing state resets $\xi_p$ for the estimator are the estimates of the observable components identified at each switching time.
In Section~\ref{sec:mapOne}, it was shown that these observable components $z_k$, $k\in\N$, comprise three subcomponents: $z_k^\cons$, $z_k^\diff$, and $z_k^\imp$.
Because of the algebraic constraints, we set $z_k^\cons = 0$, but for $z_k^\diff$ and $z_k^\imp$, we compute some appropriate estimates, denoted $\widehat{z}_k^\diff$ and $\widehat{z}_k^\imp$, respectively.
Using these estimates of the individual components, we let $\widehat{z}_k$ denoted the estimate of $z_k$, and define it as
\[
    \widehat{z}_k:= U_k^\top(z_k^\cons/\widehat{z}_k^\diff / \widehat{z}_k^\imp)
\]
with suitable compression matrix $U_k$ defined analogously as in \eqref{eq:zbardeff}. In the next two subsections, we explain how the estimates $\widehat{z}_k^\diff$ and $\widehat{z}_k^\imp$ must be calculated for the convergence result proved in Section~\ref{sec:obsConv}.


\subsubsection{Estimate the smooth part $z_k^\diff$}

Based on the discussion in Section~\ref{sec:mapOne}, it is possible to introduce a function $\mathbf{z}_{k-1}^\diff(\cdot) = {Z_{k-1}^\diff}^\top e(\cdot)$ on $(t_{k-1},t_k)$, $k \in \N$, and define an operator $\cO_{(t_{k-1},t_{k})}^\diff$ such that $\mathbf{z}_{k-1}^\diff = \cO_{(t_{k-1},t_{j})}^\diff(y^e_{(t_{pk1},t_k)})$ denotes the component of the state that can be recovered on $(t_{k-1},t_k)$ from the smooth part of the output.
We are interested in computing the estimate of the vector $z_{k}^\diff = \mathbf{z}_{k-1}^\diff(t_{k}^-)$.
The property that we require from the estimator is the following one:
\begin{eprop}
\item \label{prop:diff} For a given $\eps_k>0$, $k\in\N$, there is an estimator $\widehat{\cO}_k^\diff$ such that
\[
    \widehat{z}_k^\diff = \widehat{\cO}_k^\diff (y^e_{(t_{k-1},t_k)})
\]
has the property that
\[
\vert z_k^\diff - \widehat z_k^\diff \vert \le \eps_k |\mathbf z_{k-1}^\diff(t_{k-1}^+)|,
\]
where $|\cdot|$ denotes the Euclidian norm.
\end{eprop}
In the literature, one can find many estimation techniques for linear systems of the form \eqref{eq:zsys}.
One example of an estimator which satisfies this property is the classical \emph{Luenberger observer}.
Indeed, for $\mathbf z_{k-1}^\diff$ satisfying the equation \eqref{eq:zsys} (defined on $(t_{k-1},t_k)$), such an estimator is of the following form:
\[
\dot {\widehat {\mathbf z}}_{k-1}^\diff = (S_{k-1}^\diff - L_{k-1} R_{k-1}^\diff) \widehat{\mathbf{z}}_{k-1}^\diff + L_{k-1} y^e\ \text{on } (t_{k-1}, t_{k}),
\]
and we choose $\widehat {\mathbf z}_{k-1}^\diff (t_{k-1}^+) = 0$.
Because $(S_{k-1}^\diff, R_{k-1}^\diff)$ is an observable pair by construction, it follows from the {\em squashing lemma} \citep[Lemma~1]{PaitMors94} that for a given $\eps_k > 0$, and $\tau_{k-1} > 0$, there exists a matrix $L_{k-1}$ such that
$
   \|\me^{(S_{k-1}- L_{k-1} R_{k-1})\tau_{k-1}} \| \le \eps_{k},
$
where $\|\cdot\|$ denotes the induced matrix norm with respect to the Euclidian norm $|\cdot|$.
By setting $\widehat{z}_{k}^\diff = \widehat{\mathbf{z}}_{k-1}^\diff(t_{k}^-)$, and looking at the dynamics for $\mathbf{z}_{k-1}^\diff - \widehat{\mathbf{z}}_{k-1}^\diff$, it follows that the desired estimate $\vert z_{k}^\diff - \widehat{z}_{k}^\diff \vert \le \eps_{k} |\mathbf{z}_k^\diff(t_{k-1}^+)|$ holds.
By now, there exist many different estimation techniques for linear systems in the literature and the motivation for not fixing one particular estimation technique is to allow the possibility of several other estimators which have their own advantages.

\subsubsection{Estimate the impulsive part $z_k^\imp$}

To estimate $z_k^\imp$ from the impulsive part of the output, we write
the impulsive output of the error system \eqref{eq:impulse_error} as:
\[
   y^e[t_k] = \sum_{i=0}^{n-2} \eta_k^i \delta_{t_k}^{(i)}.
\]
Then, recalling \eqref{eq:z1imp},
\[
    z_k^\imp = U_k^{\imp^\top} \boldsymbol{\eta}_k,
\]
where $U_k^\imp$ is a compression matrix defined analogously as in \eqref{eq:U1imp} and $\boldsymbol{\eta}_k:=(\eta_k^0/\eta_k^1/\ldots/\eta_k^{n-2})$.
Recall that 
\[
   y^e[t_k] = C_k \widehat{x}_{k-1}[t_k] - y[t_k] =: C_k \sum_{i=0}^{n-2} \zeta_k^i \delta_{t_k}^{(i)} - \sum_{i=0}^{n-2} \nu_k^i \delta_{t_k}^{(i)},
\]
i.e.,
\[
   \eta_k^i = C_k \zeta_{k}^i - \nu_k^i,
\]
where $\nu_k^i$ must be obtained via measuring the impulsive part of the system's output and $\zeta_k^i$ results from running the system copy \eqref{eq:obsSyn}; in particular, $\zeta_k^i$ depends on the impulsive part of the input $u[t_k]$ as well as on the derivatives $u^{(i)}(t_k^+)$, $i=0,\ldots,n-2$, of the input immediately after the switch at $t_k$.

While the input may be known exactly (and hence $\zeta_k^i$ may be calculated analytically), obtaining $\nu_k^i$ from measurements may prove to be a very difficult task using physical sensors, because Dirac impulses do not occur in reality. One possibility to approximately determine $\nu_k^i$ is the following observation: Assume $\int$ denotes an ideal integrator, then
$
  \nu_k^0 = \left(\int y\right)(t_k+) - \left(\int y\right)(t_k^-),
$
and in general
\[
   \nu_k^i = \bigg(\underbrace{\int\int\cdots\int}_{i+1\text{ times}} y\bigg)(t_k^+) - \left(\int\int\cdots\int y\right)(t_k^-).
\]
If in reality, a Dirac impulse in $y[t_k]$ is ``smeared out'' on the interval $[t_k,t_k+\eps]$ then one gets an estimate of $\nu_k^0$ by
$
   \nu_k^0 \approx \left(\int y\right)(t_k+\eps) - \left(\int y\right)(t_k),
$
and analogously for $\nu_k^i$. The smaller $\eps$ is and the better the integrator is implemented, the closer the approximation is to the exact value $\nu_k^i$ (and also $\eta_k^i$). Hence, we may approximate $y^e[t_k]$ as 
$
   \widehat y^e[t_k] \approx \sum_{i=0}^{n-2} \widehat{\eta}_k^i \delta_{t_k}^{(i)}
$
and consequently
\[
    \widehat{z}^\imp_k :=  U_k^{\imp^\top} \widehat{\boldsymbol{\eta}}_k.
\]
\begin{eprop}
\item \label{prop:imp}For a given $\eps_k > 0$, $k\in\N$, one can obtain an approximation $\widehat{\boldsymbol{\eta}}_k$ of $\boldsymbol{\eta}_k$ such that
\[
    |\boldsymbol{\eta}_k - \widehat{\boldsymbol{\eta}}_k| \le \eps_k |\boldsymbol{\eta}_k|.
\]
\end{eprop}

\subsection{Convergence}\label{sec:obsConv}

We now show that the error correction vector $\xi_p$, $p \in \N$, computed from these estimates would make the state estimation error converge to zero, if the estimates of the observable components at each switching time are good enough, and a certain determinability assumption over intervals holds repeatedly.
To formalize this result, let us introduce the following assumption:
\begin{ass}
\item \label{ass:det}Assume that there exists a pair of subsequence in $\N$, non-decreasing and unbounded, denoted as $\{(q_i,p_i)\}_{i = 1}^\infty$ with $q_i<p_i<p_{i+1}$ and such that \eqref{eq:sysLin} is $(t_{q_i},t_{p_i}]$-determinable, i.e.\
\begin{equation}\label{eq:assDet}
   \cQ_{q_i}^{p_i} \subseteq \ker \Pi_{p_i} , \quad i = 1,2,3, \cdots.
\end{equation}

\end{ass}

Assumption~\ref{ass:det} basically requires that system \eqref{eq:sysLin} is persistently determinable, i.e., after any time instant, a determinable interval appears again eventually. At the end of these determinability intervals our observer resets the state estimate. Depending on the estimation accuracies formulated in \ref{prop:diff} and \ref{prop:imp} for individual components, the state resets make the overall estimation error sufficiently small. Afterwards the system copy runs without any continuous feedback, but its deviation from the original state is bounded by the systems dynamics. More formally, we can formulate the following qualitative convergence result:

\begin{Theorem}\label{thm:obsGenConv}
Consider the switched system \eqref{eq:sysLin} satisfying Assumption~\ref{ass:det}. For the impulsive observer \eqref{eq:obsSyn}, let
\[ 
   \xi_p = \begin{cases} \cO_{q_i}^{p_i}(\widehat{\mathbf{z}}_{q_i+1}^{p_i}) - \Xi_{q_i}^{p_i-1}(\boldsymbol{\xi}_{q_i+1}^{p_i-1}) & \text{ if } p = p_i,\ i\in\N,\\
0, & \text{ otherwise},
\end{cases}
\]
where the map $\cO_{q_i}^{p_i}$ is given by \eqref{eq:Oqp}, the estimates $\widehat{\mathbf{z}}_{q_i+1}^{p_i}=(\widehat{z}_{q_i+1}, \widehat{z}_{q_i+2},\ldots, \widehat{z}_{p_i})$ are computed as in Section~\ref{sec:zEst}, the map $\Xi_{q_i}^{p_i-1}$ is given by \eqref{eq:Xiqp-1} and $\boldsymbol{\xi}_{q_i+1}^{p_i-1}= (\xi_{q_i+1},\xi_{q_i+2},\ldots,\xi_{p_i-1})$.

For each $p \in \N$, there exists $\eps_p > 0$ such that, if $\widehat z_p$ satisfies the estimation properties \ref{prop:diff} and \ref{prop:imp} for the given $\eps_p$, then it holds that
\[
   \boxed{ \lim_{t\rightarrow \infty} |\widehat{x}(t^+) - x(t^+)| = 0.}
\]
\end{Theorem}

The proof of Theorem~\ref{thm:obsGenConv} is constructive from design viewpoint and a quantitative bound on the $\eps_p$ is computed.
Before proving this result in its generality, we highlight two special cases, whose convergence proofs form the basis for the general proof of Theorem~\ref{thm:obsGenConv}.
\begin{Definition}[Interval-wise and sliding window observer]
   Consider our general observer design as given in Theorem~\ref{thm:obsGenConv}. We call this observer \emph{interval-wise observer} if $q_{i+1}=p_i$ for all $i\in\N$, i.e.\ the determinability intervals cover the whole time axes without overlap. On the other hand, when there is the maximal possible overlap, i.e.\ $p_{i+1}=p_i+1$ for all $i\in\N$, then we call our observer \emph{sliding window observer}.
\end{Definition}

\begin{proof}[Proof of Theorem~\ref{thm:obsGenConv}]
We first observe in general that, due to Proposition~\ref{prop:mapCompRecur} and \eqref{eq:preErrC} for $p=p_i$,
\[\begin{aligned}
   e(t_{p_i}^+) &= \Pi_{p_i} ( e(t_{p_i}^-)-\xi_{p_i})\\
    &= \Pi_{p_i}\left(\cO_{q_i}^{p_i}(\mathbf{z}_{q_i+1}^{p_i}) - \Xi_{q_i}^{p_i-1}(\boldsymbol{\xi}_{q_i+1}^{p_i-1})\right.\\
    &\phantom{{}= \Pi_{p_i}~}\left.-\left(\cO_{q_i}^{p_i}(\widehat{\mathbf{z}}_{q_i+1}^{p_i}) - \Xi_{q_i}^{p_i-1}(\boldsymbol{\xi}_{q_i+1}^{p_i-1}\right)\right)\\
    &= \Pi_{p_i}\cO_{q_i}^{p_i}(\mathbf{z}_{q_i+1}^{p_i}-\widehat{\mathbf{z}}_{q_i+1}^{p_i}).
    \end{aligned}
\]
From Lemma~\ref{lem:zk-hatzk_bound} in \ref{app:lemmas}, for each $k \in \N$, there exists a constant $M^z_{k-1,k}>0$ depending on $\tau_{k-1}=t_k-t_{k-1}$, $(E_{k-1},A_{k-1},C_{k-1})$ and $(E_k,A_k,C_k)$ such that
   \begin{equation}\label{eq:boundzk}
      |z_k - \widehat{z}_k| \leq \eps_k M^z_{k-1,k} |e(t_{k-1}^+)|.
   \end{equation}
Using the notation of Theorem~\ref{thm:mapZs}, \eqref{eq:Oqp} and \eqref{eq:boundzk} yield
\begin{equation}\label{eq:recursive_error_bound}
   |e(t_{p_i}^+)| \leq \sum_{k=q_i+1}^{p_i}  \eps_k M^{\cO}_{k,p_i}  M^z_{k-1,k} |e(t_{k-1}^+)|,
\end{equation}
where
\[
  M^{\cO}_{k,p_i} := \left\| \Pi_{p_i}P_{q_i}^{p_i} \left(\prod_{j=k+1}^{p_i} G_{q_i}^j\right) F_{q_i}^k\right\|.
\]

\emph{Interval-wise observer.} For this case, we first utilize the fact that $\xi_k=0$ for $k=q_i+1,q_i+2,\ldots,p_i-1$, and by invoking Lemma~\ref{lem:diffProj}, we have for these $k$:
\[
   e(t_k^+) = \left(\prod_{j=q_i}^{k-1} \Pi_{j+1}\me^{A^\diff_j \tau_j} \right) e(t_{q_i}^+).
\]
Substitution in \eqref{eq:recursive_error_bound} gives
\[
   |e(t_{p_i}^+)| \leq c_i |e(t_{q_i}^+)| = c_i |e(t_{p_{i-1}}^+)|.
\]
The constant $c_i$ is defined as:
\[
    c_i := \sum_{k=q_i+1}^{p_i}  \eps_k M^{\cO}_{k,p_i}  M^z_{k-1,k} M^\diff_{q_i,k-2}
\]
where, for $k = q_i+2, \dots, p_i$, we let
\begin{equation}\label{eq:M^Adiff}
    M^\diff_{q_i,k-2} := \left\| \prod_{j=q_i}^{k-2} \Pi_{j+1}\me^{A^\diff_j \tau_j}\right\|
\end{equation}
and by convention, $M^\diff_{q_i,q_i-1} = 1$.
On each determinability interval $(p_{i-1},p_i]$ we therefore can chose $\eps_{p_{i-1}+1}, \eps_{p_{i-1}+2}, \ldots \eps_{p_i}$ sufficiently small such that $c_i\in(0,1)$ and $c_i$ is uniformly bounded away from $1$. We can thus conclude that in this case
\begin{equation}\label{eq:etpi-convergence}
   e(t_{p_i}^+) \underset{i\to\infty}{\longrightarrow} 0.
\end{equation}
Note that the ability to let $i$ tend to infinity (and also $t_{p_i}\to\infty$) follows from assumption \ref{ass:det}.

Next, we have for $t\in(t_k,t_{k+1})\subseteq (t_{p_{i}},t_{p_{i+1}})$:
\[Ich
   e(t^+) = \me^{A^\diff_k (t-t_k)} \left(\prod_{j=p_{i}}^{k-1} \Pi_{j+1}\me^{A^\diff_j \tau_j} \right) e(t_{p_{i}}^+),
\]
which gives
$
   |e(t^+)| \leq \overline{M}^\diff_{p_{i},p_{i+1}} |e(t_{p_{i}}^+)|,
$
where
\[
    \overline{M}^\diff_{p_{i},p_{i+1}} :=\sup_{t\in (t_{p_{i}},t_{p_{i+1}})} \left\|\me^{A^\diff_{k(t)} (t-t_{k(t)})} \left(\prod_{j=p_{i}}^{k(t)-1} \Pi_{j+1}\me^{A^\diff_j \tau_j} \right)\right\|,
\]
and $k(t)\in\{p_{i},p_{i}+1,\ldots, p_{i+1}-1\}$ is such that $t\in(t_{k(t)},t_{k(t)+1})$. We may now chose $c_{i}$ small enough (by choosing $\eps_{p_{i-1}+1}$, \ldots, $\eps_{p_{i}}$ to be sufficiently small) such that $c_{i}\overline{M}^\diff_{p_{i},p_{i+1}}$ is uniformly bounded, say by $\overline c >0$, then for all $t\in(t_{p_{i}},t_{p_{i+1}})$ we have
\[
   |e(t^+)|\leq \overline{M}^\diff_{p_{i},p_{i+1}} c_{i} |e(t_{p_{i-1}}^+)| \leq \overline c \, |e(t_{p_{i-1}}^+)|,
\]
and the convergence of $|e(t^+)|$ towards zero as $t\to\infty$ now follows from \eqref{eq:etpi-convergence}.

\emph{Sliding window observer.} By Assumption~\ref{ass:det}, the sequence $(q_i)_{i\in\N}$ is nondecreasing and unbounded, i.e., for sufficiently large $i$ we have $q_i\geq p_1$. By using the relation $p_{i+1}=p_i+1$, we can rewrite \eqref{eq:recursive_error_bound} as
\[
    |e(t_{p_i}^+)| \leq \sum_{k=q_i+1}^{p_i}  \eps_k M^{\cO}_{k,p_i}  M^z_{k-1,k} |e(t_{p_{i-p_i+k-1}}^+)|
\]
for sufficiently large $i$. Let
\[
   c_i := (p_i-q_i) \max_{k=q_i+1,\ldots,p_i}\eps_k M^\cO_{k,p_i} M^z_{k-1,k}.
\]
Since $(q_i)_{i\in\M}$ is non-decreasing and unbounded, for any fixed $k$, there are only finitely many indices $i$ such that $c_i$ depends on $\eps_k$. Hence for sufficiently small $\eps_k$ we have $c_i\in(0,1)$ uniformly bounded away from $1$ and
\[
    |e(t_{p_i}^+)| \leq \frac{c_i}{p_i-q_i} \sum_{j=1}^{p_i-q_i} |e(t_{p_{i-j}}^+)|.
\]
Applying Lemma~\ref{lem:convSeq} from the \ref{app:lemmas} results in
\begin{equation}\label{eq:slide_err_ti}
   |e(t_{p_i}^+)|\underset{i\to\infty}{\longrightarrow} 0.
\end{equation}
Next, for any $t\in (t_{p_i},t_{p_{i}+1})$, we have $e(t^+) = \me^{A^\diff_{p_i}(t-t_{p_i})} e(t_{p_i}^+)$ and hence
\[
   |e(t^+)| \leq M^\diff_{p_i} |e(t_{p_i}^+)| \leq  M^\diff_{p_i} c_i \frac{1}{p_i-q_i}\sum_{j=1}^{p_i-q_i} |e(t_{p_{i-j}}^+)|,
\]
where
$
   M^\diff_{p_i} = \sup_{s\in(0,\tau_{p_i})} \left\|\me^{A^\diff_{p_i}s}\right\|.
$
For sufficiently small $\eps_k$, we can ensure that $M^\diff_{p_i} c_i$ is uniformly bounded, say by $\overline c >0$. Furthermore, for any $\eps>0$ there exists an index $i_{\eps}$ such that $|e(t_p^+)|\leq \eps$ for all $p\geq q_{i_\eps}$, hence for all $t\in(t_{p_i},t_{p_i+1})$ and all $i\geq i_\eps$ we have
\[
   |e(t^+)| \leq \overline{c} \, \frac{1}{p_i-q_i}\sum_{j=1}^{p_i-q_i} |e(t_{p_{i-j}}^+)| \leq \overline{c} \eps,
\]
This shows convergence of $e(t^+)$ towards $0$ for $t\to\infty$.

\emph{The general case.} We now combine the proof ideas from the interval-wise and sliding window observer to also prove the general case.
To this end, for a fixed $i \in \N$, introduce the function $h(i,\cdot): \{q_i, \dots, p_i-1\} \rightarrow \N$ such that\footnote{
For the interval-wise observer, $h(i,k)=p_{i-1} = q_i$. For sliding-window observer, where $p_{i-1} = p_i - 1$, we had $h(i,k)= k = p_i-p_i + k = p_{i-p_i+k}$.}
\[
h(i,k) = \max \{j \, \vert \, p_j \le k\}
\]
and let the set $J_i$ be the range of $h(i,\cdot)$, that is, $J_i := \{h(i,k), k=q_i, \dots, p_i-1\}$.
For each $k = q_i, \dots, p_i - 1$, it holds that, recalling \eqref{eq:M^Adiff},
\[
\vert e(t_{k}^+) \vert \le M_{p_{h(i,k)},k-1}^\diff |e(t_{p_{h(i,k)}}^+)|
\]
and \eqref{eq:recursive_error_bound} becomes
\[
\vert e(t_{p_i}^+) \vert \le \!\!\!\sum_{k=q_i+1}^{p_i} \!\!\!\eps_{k} M_{k,p_i}^{\cO} M_{k-1,k}^z M_{p_{h(i,k-1)},k-2}^\diff \, \vert e(t_{p_{h(i,k-1)}}) \vert .
\]
Let $\vert J_i \vert$ denote the cardinality of $J_i$, and let
\[
c_i := \left \vert J_i \right \vert \max_{k=q_i+1,\ldots,p_i}\eps_k M^\cO_{k,p_i} M^z_{k-1,k} M_{p_{h(i,k-1)},k-2},
\]
then by choosing $\eps_k$ sufficiently small, we again have $c_i \in (0,1)$ uniformly bounded away from 1, and
\[
    |e(t_{p_i}^+)| \leq \frac{c_i}{\vert J_i\vert} \sum_{j=1}^{\vert J_i \vert} |e(t_{p_{i-j}}^+)|.
\]
Once again, it follows from Lemma~\ref{lem:convSeq} in \ref{app:lemmas} that $|e(t_{p_i}^+)| \to 0$ as $i \rightarrow \infty$. To show that, $|e(t^+)|$ converges to zero for $t \in (t_{p_i},t_{p_{i+1}})$, one can follow exactly the same arguments as in the case of interval-wise observer to concluded that $e(t^+)\to 0$ as $t\to\infty$.
\end{proof}

\begin{Remark}[Convergence of impulsive part {$\widehat x[t_p] - x[t_p]$}]
It was already observed that
  $
     e[t_p] = -\sum_{i=0}^{n-2} (E_p^\imp)^{i+1} (e(t_p^-) - \xi_p) \delta_{t_p}^{(i)}.
  $
From Theorem~\ref{thm:obsGenConv}, we have $e(t_p^-)$ converging to zero for large $p \in \N$. Since $\xi_p$ is by construction an estimate of $e(t_p^-)$ which gets closer and closer to the real value of $e(t_p^-)$ as $p$ gets large, it follows that the coefficients multiplying Dirac impulses get smaller with time.
\end{Remark}


\section{Simulations}\label{sec:sim}
For the simulation of a sliding window observer, we refer to our conference paper \citep{TanwTren13}. The simulations for the interval-wise observer are now presented for the example considered in Section~\ref{sec:example}. As a known input to the system, we chose
$
   u(t) = 1 + \sin(t).
$
The corresponding output of the switched system \eqref{eq:sysLin} with initial condition $x(t_0) = [1,\ 3/2,\ 2,\ 5/2]^\top$ is shown in Figure~\ref{fig:output}.
\begin{figure}[htb]
\centering
  \includegraphics[width=0.375\textwidth]{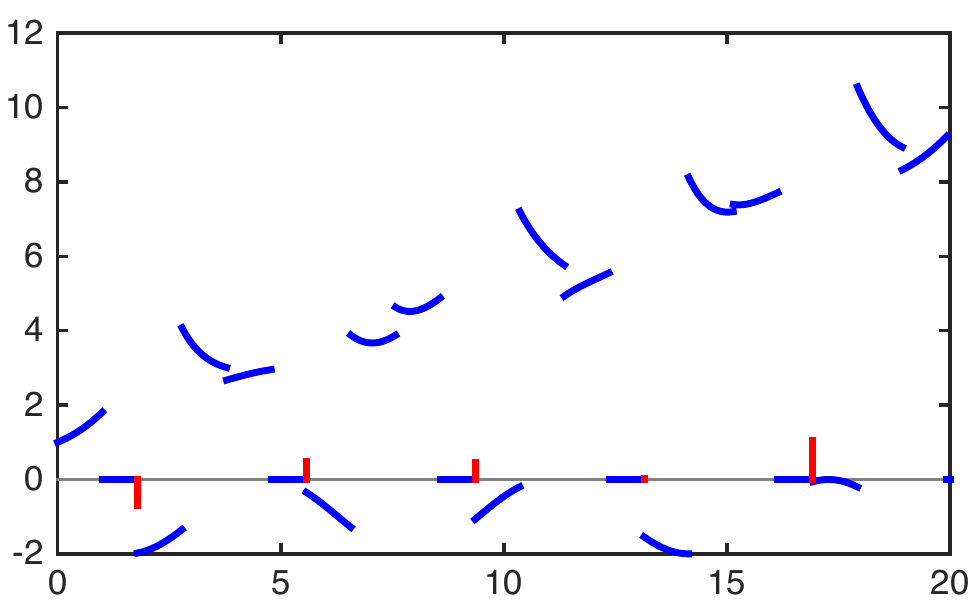}
  \caption{Output $y=y^f_\D + y[\cdot]$ of \eqref{eq:sysLin} with regular part $y^f$ (blue) and indication of Dirac impulses in $y[\cdot]$ (red).}\label{fig:output}
\end{figure}

For our observer design, we estimate $z^{\diff}_p$ using a classical Luenberger observer with 
$
   L_0 = [1/4,3/8]^\top,\quad L_3 = [3/8, 1/4]^\top,
$
and we add ten percent measurement noise in the impulsive part of $y[\cdot]$. The system copy starts with zero initial value. The resulting estimation errors (without the Dirac impulses) are shown in Figure~\ref{fig:errors}.
\begin{figure}[hbt]
  \includegraphics[width=0.475\textwidth]{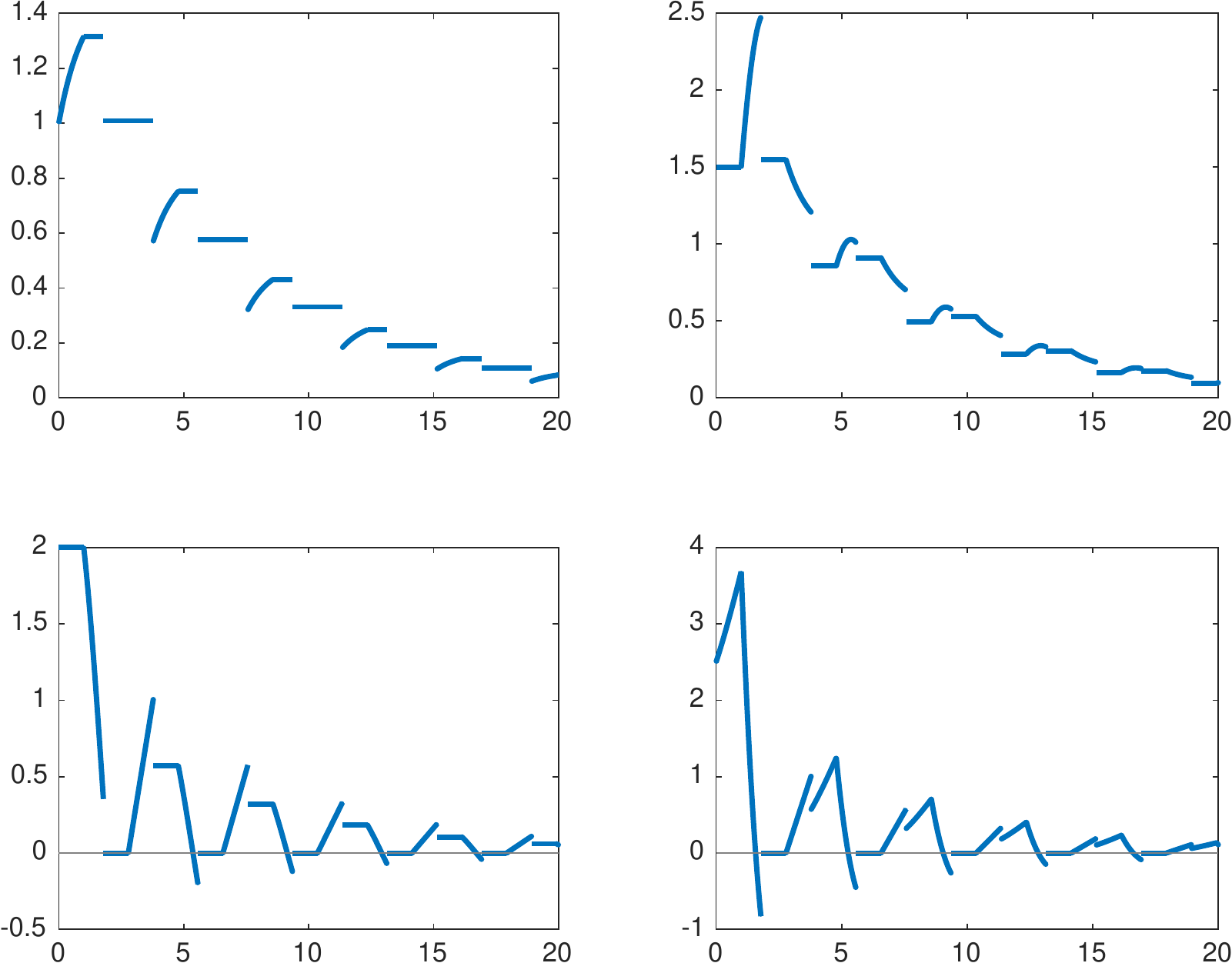}
  \caption{Estimation errors (without Dirac impulses), $x_1-\widehat{x}_1$ upper left figure, $x_2-\widehat{x}_2$ upper right figure, $x_3-\widehat{x}_3$ lower left figure, $x_4-\widehat{x}_4$ lower right figure.}\label{fig:errors}
\end{figure}
Clearly, the estimation error converges (slowly) to zero. Note that the convergence can be accelerated significantly by using a more aggressive gain matrix in Luenberger observers. However, there exists a lower bound on the gain matrices which is determined by the length of the determinability intervals.

\section{Conclusions}\label{sec:conc}
The paper considered the problem of state estimation in switched linear DAEs.
The notion of determinability studied in this paper relates to the reconstruction of the state value at some time by processing outputs and inputs over an interval.
This does not necessarily require observability of the initial state, or the individual subsystems.
The geometric characterization of determinability is then used for synthesis of a class of state estimators.
In contrast to classical estimation techniques which require continuous output injection, in our approach the estimator is reset at some discrete time instants after processing the external measurements over an interval.

For future work, we are interested in developing state estimators which only require the property of detectability from system.
Our preliminary results on detectability of switched DAEs have appeared in \citep{TanwTren15}.
Similar concepts have been used for studying the notion of controllability in switched DAEs \citep{KustRupp15} and a duality result is also available \citep{KustTren15pp}. It would be interesting to investigate if such ideas can be used for designing stabilizing controllers.

\appendix
\section{Detailed Proofs}\label{app:proofs}

\begin{proof}[Proof of Proposition~\ref{prop:fwdObsSing}]
The above discussion already shows ``$\supseteq$''. To show the converse subspace inclusion, let $e_1 \in \cW_1$. We need to show that there exists a solution $e$ of \eqref{eq:sysHomLin} such that $e(t_1^-) = e_1$ and the resulting output is such that $y^e_{(t_0,t_1]} \equiv 0$. For that let $e_0:= \me^{-A_0^\diff (t_1-t_0)}e_1$. Since $e_1\in\fC_0$ and $\fC_0$ is $A^\diff_0$-invariant, it follows that $e_0\in\fC_0$. Lemma~\ref{lem:diffProj} now yields that \eqref{eq:sysHomLin} with initial condition $e(t_0^+) = e_0$ has a solution given by $e(t)=e^{A_0^\diff (t-t_0)}e_0$ on $(t_0,t_1)$, and in particular, $e(t_1^-) = e_1$. Furthermore, $e_1\in\ker O^\diff_0$ and $A_0^\diff$-invariance of $\ker O^\diff_0$ implies that $e(t)\in\ker O^\diff_0$ for $t \in (t_0,t_1)$, i.e., $y^e_{(t_0,t_1)}=0$. Finally, $e(t_1^-)=e_1\in\ker O_0^\imp$, hence $y^e[t_1]=0$ by Lemma~\ref{lem:impulses}.
\end{proof}

\begin{proof}[Proof of Theorem~\ref{thm:DAEDetMult}]
We first prove \eqref{eq:impFwdObsMult}.

($\supseteq$).
To prove this inclusion, we consider the solution $e$ that satisfies \eqref{eq:sysHomLin} on $(t_q,t_p]$ and the corresponding output $y^e$ is such that $y^e_{(t_q,t_p]}=0$. From Proposition~\ref{prop:fwdObsSing}, we have $e(t_{q+1}^-) \in \cW_{q+1} = \cQ_q^{q+1}$.
Using an induction based argument, it is shown that $e(t_k^-)\in\cQ_q^k$ for all $k = q+1, q+2, \ldots, p$.
Assume that the desired relation holds for some $k \ge q+1$. As $e(t_{k+1}^-)= \me^{A_k^\diff \tau_k} \Pi_k e(t_k^-)$, it follows that $e(t_{k+1}^-)\in \me^{A_k^\diff \tau_k} \Pi_k \cQ_q^k$. Also, the application of Proposition~\ref{prop:fwdObsSing} on the interval $(t_k,t_{k+1}]$ yields that $e(t_{k+1}^-)\in \cW_{k+1}$. Hence we have shown that $e(t_{k+1}^-) \in \cW_{k+1} \cap \me^{A_k^\diff \tau_k} \Pi_k \cQ_q^k = \cQ_q^{k+1}$.

($\subseteq$).
For each $e_p \in \cQ_q^p$, we want to construct a solution $e$ of \eqref{eq:sysHomLin} over the interval $(t_q,t_p]$ such that $e(t_p^-) = e_p$ and $e(t_k^-) \in \cW_k$, for $k=q+1,q+2,\ldots,p$. From Proposition~\ref{prop:fwdObsSing}, it would then follow that $y^e_{(t_{k-1},t_k]} = 0$ and the claim is shown.
To construct the desired solution, let $e_p\in \cQ_q^p$ and chose $e_k \in \cQ_q^k$ for $k=p-1, p-2, \ldots, q+1$ so that
\[
 \me^{A_k^\diff \tau_k} \Pi_k e_k = e_{k+1},
\]
which is possible due to the definition of $\cQ_q^{k+1}$. Finally let $e_q^+:= \me^{-A^\diff_q \tau_q} e_{q+1}$. By construction, $e_{q+1}\in\fC_{q}$ and hence $e_q^+\in\fC_q$ as well. Hence there exists a local\footnote{
Note that there may not exist a global solution with the initial condition $e(t_q^+)=e_q^+$, as already shown by example in Remark~\ref{rem:local_sol}.}
solution $e$ of \eqref{eq:sysHomLin} on $(t_q,t_p]$ with $e(t_q^+)=e_q^+$. Furthermore, Lemmas~\ref{lem:consProj} and \ref{lem:diffProj} yield $e(t_{k+1}^-) = \me^{A_k^\diff \tau_k} \Pi_k e(t_k^-)$ for all $k=q+2,q+3,\ldots,p-1$ as well as $e(t_{q+1}^-)=\me^{A_q^\diff\tau_q} e(t_q^+) = e_{q+1}$. Inductively, we can now conclude that $e(t_k^-) = e_k \in \cQ_q^k \subseteq \cW_k$ for all $k=q+1,q+2,\ldots p$ which concludes this proof step.

Finally,  due to Lemma~\ref{lem:consProj}, the condition \eqref{eq:CharacDetDAE} is a characterization for determinability of \eqref{eq:sysHomLin} and Proposition~\ref{prop:fwdObsZero} yields the same determinability characterization for the inhomogeneous switched DAE \eqref{eq:sysLin}.
\end{proof}

\section{Lemmas Used in Derivations}\label{app:lemmas}
\begin{Lemma}\label{lem:reduced_obs}
   Consider the ODE
   \[
       \dot{x} = Ax,\quad y = Cx
   \]
   for some $A\in\R^{n\times n}$ and $C\in\R^{\dy\times n}$, and chose a matrix $Z$ with orthonormal columns such that
   \[
      \im Z = \im [C / CA / CA^2 / \ldots/ C A^{n-1}]^\top.
   \]
   Then for any solution $x$ of $\dot{x}=Ax$ we have that $z = Z^\top x$ is a solution of the observable ODE
   \[
      \dot{z} = Z^\top A Z z,\quad y = C Z z.
   \]
\end{Lemma}
\begin{proof}
   This is a simple consequence from the well known Kalman observability decomposition.
\end{proof}
\begin{Lemma}\label{lem:zk-hatzk_bound}
   Consider the switched DAE \eqref{eq:sysLin} together with the impulsive observer \eqref{eq:obsSyn} and corresponding error dynamics \eqref{eq:errDyn}. For $k\in\N$ let $z_k$ and $\widehat{z}_k$ be given as in Section~\ref{sec:mapOne} and Section~\ref{sec:zEst}. If for some $\eps_k > 0$, the estimation properties \ref{prop:diff} and \ref{prop:imp} hold, then there exists a constant $M^z_{k-1,k}>0$ depending on $\tau_{k-1}=t_k-t_{k-1}$, $(E_{k-1},A_{k-1},C_{k-1})$ and $(E_k,A_k,C_k)$ such that
   \begin{equation}\label{eq:boundzk-2}
      |z_k - \widehat{z}_k| \leq \eps_k M^z_{k-1,k} |e(t_{k-1}^+)|.
   \end{equation}
\end{Lemma}
\begin{proof}
  Since $z_k-\widehat{z}_k = U_k^\top (0 / z_k^\diff -\widehat{z}_k^\diff / z_k^\imp - \widehat{z}_k^\imp)$, it suffices to consider the differences $z_k^\diff -\widehat{z}_k^\diff$ and $z_k^\imp - \widehat{z}_k^\imp$ individually.
  
  \emph{Bound for the difference $|z_k^\diff -\widehat{z}_k^\diff|$.} Invoking \ref{prop:diff} and the definition of $\mathbf{z}^\diff_{k-1}$ implies
  \[\begin{aligned}
     \left|z_k^\diff -\widehat{z}_k^\diff\right| &\leq \eps_k \left|\mathbf{z}_k^\diff(t_{k-1}^+)\right| = \eps_k \left|{Z_{k-1}^\diff}^\top e(t_{k-1}^+)\right|\\
     &\leq \eps_k \left\|{Z_{k-1}^\diff}^\top\right\| \left| e(t_{k-1}^+)\right|,
     \end{aligned}
  \]
  which is the desired bound.

  \emph{Bound for the difference $|z_k^\imp -\widehat{z}_k^\imp|$.} We first observe that
  \[
     y^e[t_k] = C_k (\widehat{x}_{k-1}[t_k]-x[t_k]) = C_k(\widehat{x}_{k-1}[t_k] - \widehat{x}_k[t_k] - e[t_k]).
  \]
From Lemma~\ref{lem:impulses}, it is easily seen that,
  \[
     \widehat{x}_{k-1}[t_k] - \widehat{x}_k[t_k] = -\sum_{i=0}^{n-2} (E_k^\imp)^{i+1} \xi_k \delta^{(i)}
  \]
  and
  \[
     e[t_k] = -\sum_{i=0}^{n-2} (E_k^\imp)^{i+1} (e(t_k^-) - \xi_k) \delta_{t_k}^{(i)}.
  \]
  Hence, $y^e[t_k]= \sum_{i=0}^{n-2} \eta_k^i \delta_{t_k}^{(i)}$ with
  \[
     \eta_k^i = - C_k(E_k^\imp)^{i+1} e(t_k^-) = - C_k(E_k^\imp)^{i+1} \me^{A^\diff_{k-1} \tau_{k-1}} e(t_{k-1}^+),
  \]
  where we invoked Lemma~\ref{lem:diffProj}. Combining this with \ref{prop:imp}, we obtain the desired bound:
  \[\begin{aligned}
     \left|z_k^\imp -\widehat{z}_k^\imp\right| &= \left|{U_k^\imp}^\top (\boldsymbol{\eta}_k - \widehat{\boldsymbol{\eta}}_k)\right|
     \leq \eps_k \left\|{U_k^\imp}^\top\right\| \left|\boldsymbol{\eta}_k\right|\\
     &\leq \eps_k M_{k-1,k}^\imp \left|e(t_{k-1}^+)\right|,
\end{aligned}
\]
 where
 $
   M_{k-1,k}^\imp := \left\|{U_k^\imp}^\top\right\| \left\|{O_k^\imp} \me^{A_{k-1}^\diff \tau_{k-1}}\right\|.
 $
 Altogether, we have \eqref{eq:boundzk-2} with
 \[
    M^z_{k-1,k} = \left\|U_k^\top\right\| \left| \begin{pmatrix}\|Z_{k-1}^\diff\|\\ M^\imp_{k-1,k}\end{pmatrix} \right|.
 \]
 \end{proof}

\begin{Lemma}\label{lem:convSeq}
Consider a sequence $(a_k)_{k\in\N}$ in $\R^m$, $m\in\N$, where
\begin{equation}\label{eq:average_bound}
  |a_k| \le \frac{\alpha}{n_k} \sum_{j=1}^{n_k} |a_{k-j}|
\end{equation}
for some $\alpha\in(0,1)$, and the positive integers $n_k \le k$ are such that the sequence $(k-n_k)_{k \in \N}$ is nondecreasing and unbounded.
It then holds that $\lim_{k\rightarrow \infty} |a_k| = 0$.
\end{Lemma}
\begin{proof}
We first extract a subsequence $(a_{k_i})_{i\in \N}$ such that
$
k_i \le k_{i+1} - n_{k_{i+1}}
$
which is possible due to assumptions on $(k-n_k)_{k \in \N}$.
We will show that
\begin{equation}\label{eq:indIneqMain}
\max_{k_i \le j < k_{i+1}} \vert a_j \vert \le \alpha \, \max_{k_{i-1} \le j < k_i } \vert a_j \vert
\end{equation}
which implies that $|a_k|$ converges to zero as $k \rightarrow \infty$.

From \eqref{eq:average_bound}, it follows that
$
|a_{k_i}| \le \alpha \max_{k_i-n_{k_i} \le j < k_i} |a_j|,
$
and knowing that $k_{i-1} \le k_i - n_{k_i}$, we have
\[
    |a_{k_i}| \le \alpha \max_{k_{i-1} \le j < k_i} |a_j|.
\]
We will now show inductively that for all $\ell\in\N$
\[
   |a_{k_i+\ell}| \le \alpha \max_{k_{i-1} \le j < k_i} |a_j|,
\]
from which \eqref{eq:indIneqMain} follows. 

For this, it is first observed that, for all $\ell >0$,
\[
   k_i + \ell - n_{k_i + \ell} \geq k_i - n_{k_i} \geq k_{i-1},
\]
where the first inequality is due to the nondecreasing assumption on $(k-n_k)_{k \in \N}$, and the second inequality results from the definition of $(a_{k_i})_{i \in \N}$.
To invoke the induction argument, we assume
\[
   |a_{j}| \le \max_{k_{i-1} \le j < k_i} |a_j|\quad\text{for }k_{i-1}\leq j \leq k_i+\ell-1
\]
and arrive at
\[\begin{aligned}
   |a_{k_i+\ell}| &\leq \alpha \max_{k_i+\ell - n_{k_i+\ell}\leq j\leq k_i+\ell-1} |\alpha_j|\\
     &\leq \alpha\max_{k_{i-1}\leq j \leq k_i + \ell - 1} |\alpha_j|\\
     & \leq \alpha\max_{k_{i-1}\leq j \leq k_i} |\alpha_j|.
   \end{aligned}
\]
\end{proof}

\section*{References}


\begin{thebibliography}{34}
\expandafter\ifx\csname natexlab\endcsname\relax\def\natexlab#1{#1}\fi
\expandafter\ifx\csname url\endcsname\relax
  \def\url#1{\texttt{#1}}\fi
\expandafter\ifx\csname urlprefix\endcsname\relax\def\urlprefix{URL }\fi

\bibitem[{Babaali and Pappas(2005)}]{BabaPapp05}
Babaali, M., Pappas, G.~J., 2005. Observability of switched linear systems in
  continuous time. In: Hybrid Systems: Computation and Control. Vol. 3414 of
  LNCS. Springer, Berlin, pp. 103--117.

\bibitem[{Balluchi et~al.(2003)Balluchi, Benvenuti, Di~Benedetto, and
  Sangiovanni-Vincentelli}]{BallBenv03}
Balluchi, A., Benvenuti, L., Di~Benedetto, M.~D., Sangiovanni-Vincentelli, A.,
  2003. Observability for hybrid systems. In: Proc. 42nd~{IEEE} Conf. Decis.
  Control, Hawaii, USA. Vol.~2. pp. 1159--1164.

\bibitem[{Berger et~al.(2012)Berger, Ilchmann, and Trenn}]{BergIlch12a}
Berger, T., Ilchmann, A., Trenn, S., 2012. The quasi-{W}eierstra{\ss} form for
  regular matrix pencils. Linear Algebra Appl. 436~(10), 4052--4069.

\bibitem[{Berger and Reis(2015)}]{BergReis15ppb}
Berger, T., Reis, T., 2015. Observers and dynamic controllers for linear
  differential-algebraic systems, submitted for publication.

\bibitem[{Berger et~al.(2016)Berger, Reis, and Trenn}]{BergReis16a}
Berger, T., Reis, T., Trenn, S., 2016. Observability of linear
  differential-algebraic systems. In: Ilchmann, A., Reis, T. (Eds.), Surveys in
  Differential-Algebraic Equations IV. Differential-Algebraic Equations Forum.
  Springer-Verlag, Berlin-Heidelberg, to appear.

\bibitem[{Bobinyec et~al.(2011)Bobinyec, Campbell, and Kunkel}]{BobiCamp11}
Bobinyec, K., Campbell, S.~L., Kunkel, P., 2011. Full order observers for
  linear {DAEs}. In: Proc. 50th~{IEEE} Conf. Decis. Control and European
  Control Conference ECC 2011, Orlando, USA. pp. 4011 -- 4016.

\bibitem[{Bobinyec and Campbell(2014)}]{BobiCamp14}
Bobinyec, K.~S., Campbell, S.~L., 2014. Linear differential algebraic equations
  and observers. In: Ilchmann, A., Reis, T. (Eds.), Surveys in
  Differential-Algebraic Equations II. Differential-Algebraic Equations Forum.
  Springer-Verlag, Berlin-Heidelberg, pp. 1--67.

\bibitem[{Dai(1989)}]{Dai89a}
Dai, L., 1989. Singular Control Systems. No. 118 in Lecture Notes in Control
  and Information Sciences. Springer-Verlag, Berlin.

\bibitem[{Darouach(2012)}]{Daro12}
Darouach, M., 2012. On the functional observers for linear descriptor systems.
  Syst. Control Lett. 61, 427 -- 434.

\bibitem[{Fahmy and O'Reilly(1989)}]{FahmOrei89}
Fahmy, M.~M., O'Reilly, J., 1989. Observers for descriptor systems. Int. J.
  Control 49, 2013--2028.

\bibitem[{Goebel et~al.(2009)Goebel, Sanfelice, and Teel}]{GoebSanf09}
Goebel, R., Sanfelice, R.~G., Teel, A.~R., 2009. Hybrid dynamical systems.
  {IEEE} Control Systems Magazine 29~(2), 28--93.

\bibitem[{Gross et~al.(2014)Gross, Trenn, and Wirsen}]{GrosTren14}
Gross, T.~B., Trenn, S., Wirsen, A., 2014. Topological solvability and index
  characterizations for a common {DAE} power system model. In: Proc. 2014 IEEE
  Conf. Control Applications (CCA). IEEE, pp. 9--14.

\bibitem[{Hespanha et~al.(2005)Hespanha, Liberzon, Angeli, and
  Sontag}]{HespLibe05}
Hespanha, J., Liberzon, D., Angeli, D., Sontag, E., 2005. Nonlinear
  norm-observability notions and stability of switched systems. ieeetac 50~(2),
  154--168.

\bibitem[{K{\"u}sters et~al.(2015)K{\"u}sters, Ruppert, and Trenn}]{KustRupp15}
K{\"u}sters, F., Ruppert, M. G.-M., Trenn, S., 2015. Controllability of
  switched differential-algebraic equations. Syst. Control Lett. 78~(0), 32 --
  39.

\bibitem[{K{\"u}sters and Trenn(2015)}]{KustTren15pp}
K{\"u}sters, F., Trenn, S., 2015. Duality of switched {DAE}s, submitted for
  publication.

\bibitem[{Pait and Morse(1994)}]{PaitMors94}
Pait, F., Morse, S., 1994. A cyclic switching strategy for parameter-adaptive
  control. {IEEE} Trans. Autom. Control 39~(6), 1172 -- 1183.

\bibitem[{Petreczky et~al.(2015)Petreczky, Tanwani, and Trenn}]{PetrTanw15}
Petreczky, M., Tanwani, A., Trenn, S., 2015. Observability of switched linear
  systems. In: Djemai, M., Defoort, M. (Eds.), Hybrid Dynamical Systems. Vol.
  457 of Lecture Notes in Control and Information Sciences. Springer-Verlag,
  pp. 205--240.

\bibitem[{Shim and Tanwani(2014)}]{ShimTanw14}
Shim, H., Tanwani, A., 2014. Hybrid-type observer design based on a sufficient
  condition for observability in switched nonlinear systems. Int. J. Robust \&
  Nonlinear Control: Special Issue on High Gain Observers and Nonlinear Output
  Feedback Control 24~(6), 1064 -- 1089.

\bibitem[{Sontag(1998)}]{Sont98a}
Sontag, E.~D., 1998. Mathematical Control Theory: Deterministic Finite
  Dimensional Systems, 2nd Edition. Springer-Verlag, New York.

\bibitem[{Sun et~al.(2002)Sun, Ge, and Lee}]{SunGe02}
Sun, Z., Ge, S.~S., Lee, T.~H., 2002. Controllability and reachability criteria
  for switched linear systems. Automatica 38, 775--786.

\bibitem[{Tanwani(2011)}]{TanwPhD11}
Tanwani, A., 2011. Invertibility and observability of switched systems with
  inputs and outputs. Ph.D. thesis, University of Illinois at Urbana-Champaign,
  IL, USA, available online: {\tt http://hdl.handle.net/2142/29755}.

\bibitem[{Tanwani et~al.(2014)Tanwani, Brogliato, and Prieur}]{TanwBrog14}
Tanwani, A., Brogliato, B., Prieur, C., 2014. Stability and observer design for
  multivalued {L}ur'e systems with non-monotone, time-varying nonlinearities
  and state jumps. {SIAM} J. Control \& Optim. 52~(6), 3639--3672.
  
  \bibitem[{Tanwani et~al.(2016)Tanwani, Brogliato, and
  Prieur}]{TanwBrog16}
Tanwani, A., Brogliato, B., Prieur, C., 2016. Observer design 
for unilaterally constrained {L}agrangian systems: A passivity-based approach.
{IEEE} Trans. Autom. Control 61~(9), 2386--2401.

\bibitem[{Tanwani et~al.(2013)Tanwani, Shim, and
  Liberzon}]{TanwShim13}
Tanwani, A., Shim, H., Liberzon, D., 2013. Observability for
  switched linear systems: Characterization and observer design. {IEEE} Trans.
  Autom. Control 58~(4), 891--904.

\bibitem[{Tanwani et~al.(2015)Tanwani, Shim, and Liberzon}]{TanwShim15}
Tanwani, A., Shim, H., Liberzon, D., 2015. Hybrid Dynamical Systems. Vol. 457
  of Lecture Notes in Control and Information Sciences. Springer, Ch. Observer
  Design for Switched Linear Systems with State Jumps, pp. 179 -- 203.

\bibitem[{Tanwani and Trenn(2010)}]{TanwTren10}
Tanwani, A., Trenn, S., 2010. On observability of switched
  differential-algebraic equations. In: Proc. 49th~{IEEE} Conf. Decis. Control,
  Atlanta, USA. pp. 5656--5661.

\bibitem[{Tanwani and Trenn(2012)}]{TanwTren12}
Tanwani, A., Trenn, S., 2012. Observability of switched differential-algebraic
  equations for general switching signals. In: Proc. 51st~{IEEE} Conf. Decis.
  Control, Maui, USA. pp. 2648--2653.

\bibitem[{Tanwani and Trenn(2013)}]{TanwTren13}
Tanwani, A., Trenn, S., 2013. An observer for switched differential-algebraic
  equations based on geometric characterization of observability. In: Proc.
  52nd~{IEEE} Conf. Decis. Control, Florence, Italy. pp. 5981--5986.

\bibitem[{Tanwani and Trenn(2015)}]{TanwTren15}
Tanwani, A., Trenn, S., 2015. On detectability of switched linear
  differential-algebraic equations. In: Proc. 54th~{IEEE} Conf. Decis. Control,
  Osaka, Japan. pp. 2957--2962.

\bibitem[{Trenn(2009)}]{Tren09d}
Trenn, S., 2009. Distributional differential algebraic equations. Ph.D. thesis,
  Institut f{\"u}r Mathematik, Technische Universit{\"a}t Ilmenau,
  Universit{\"a}tsverlag Ilmenau, Germany.
\urlprefix\url{http://www.db-thueringen.de/servlets/DocumentServlet?id=13581}

\bibitem[{Trenn(2012)}]{Tren12}
Trenn, S., 2012. Switched differential algebraic equations. In: Vasca, F.,
  Iannelli, L. (Eds.), Dynamics and Control of Switched Electronic Systems -
  Advanced Perspectives for Modeling, Simulation and Control of Power
  Converters. Springer-Verlag, London, Ch.~6, pp. 189--216.

\bibitem[{Trenn(2013)}]{Tren13a}
Trenn, S., 2013. Solution concepts for linear {DAE}s: a survey. In: Ilchmann,
  A., Reis, T. (Eds.), Surveys in Differential-Algebraic Equations I.
  Differential-Algebraic Equations Forum. Springer-Verlag, Berlin-Heidelberg,
  pp. 137--172.

\bibitem[{Vidal et~al.(2003)Vidal, Chiuso, Soatto, and Sastry}]{VidaChiu03}
Vidal, R., Chiuso, A., Soatto, S., Sastry, S., 2003. Observability of linear
  hybrid systems. In: Hybrid Systems: Computation and Control. Vol. 2623 of
  Lecture Notes in Computer Science. Springer, Berlin, pp. 526--539.

\bibitem[{Xie and Wang(2003)}]{XieWang03}
Xie, G., Wang, L., 2003. Controllability and stabilizability of switched
  linear-systems. Syst. Control Lett. 48~(2), 135--155.

\end{thebibliography}

\end{document}